\documentclass[format=acmsmall, review=false]{acmart}
\usepackage{amsmath}
\usepackage{subfig}
\usepackage{tikz-cd}
\usepackage{acm-ec-25}
\usepackage{booktabs}
\usepackage[ruled]{algorithm2e}
\usepackage{amsthm}

\newtheorem{lemma}{Lemma}
\newtheorem{remark}{Remark}
        
\newtheorem{claim}{Claim}

\SetAlFnt{\small}
\SetAlCapFnt{\small}
\SetAlCapNameFnt{\small}
\SetAlCapHSkip{0pt}
\IncMargin{-\parindent}
\setcitestyle{authoryear}
\begin{abstract}
Implementation theory has made significant advances in characterizing which social choice functions can be implemented in Nash equilibrium, but these results typically assume sophisticated strategic reasoning by agents. However, evidence exists to show that agents frequently cannot perform such reasoning. In this paper, we present a finite mechanism which fully implements Maskin-monotonic social choice functions as the outcome of the unique correlated equilibrium of the induced game. Due to the results in \cite{HartMasColell2000}, this yields that even when agents use a simple adaptive heuristic like regret minimization rather than computing equilibrium strategies, the designer can expect to implement the SCF correctly. We demonstrate the mechanism's effectiveness through simulations in a bilateral trade environment, where agents using regret matching converge to the desired outcomes despite having no knowledge of others' preferences or the equilibrium structure. The mechanism does not use integer games or modulo games.
\end{abstract}

\begin{document}

\title{Correlated equilibrium implementation: Navigating toward social optima with learning dynamics}
\author{Soumen Banerjee and Yi-Chun Chen and Yifei Sun}

% Title page for title and abstract only.
\begin{titlepage}

\maketitle

% Optionally include a table of contents
\vspace{1cm}
\setcounter{tocdepth}{2} % adjust to 1 if desired
\tableofcontents

\end{titlepage}

% Paper body

\section{Introduction}

Implementation theory has made significant strides in characterizing
implementable social choice functions. \cite{maskin99} establishes that any
social choice function satisfying monotonicity can be implemented in Nash
equilibrium. However, this and many subsequent results rely on agents being
able to compute and play Nash equilibria (and indirectly, best responses) -
an assumption that has faced increasing scrutiny from both theoretical and
empirical perspectives (see for instance \cite{GoereeLouis2021}, \cite%
{Li2017}).

This has led to a growing interest in implementation under more realistic
behavioral assumptions. A particularly promising direction comes from the
demonstration in \cite{HartMasColell2000} that simple regret minimization
learning leads to correlated equilibria rather than Nash equilibria.%
\footnote{%
The Pluribus robot designed by researchers at Carnegie Mellon university,
which beat the best professional human players at no-limit Texas hold 'em
poker also uses a variation of regret minimization, suggesting that the
algorithm is not just simple, but also capable.} In addition, other adaptive
procedures such as calibrated learning (due to \cite{foster1997calibrated})
and universal conditional consistency (due to \cite{fudenberg1999conditional}%
) have also been shown to approach correlated equilibria.\footnote{\cite%
{stoltz2007learning} demonstrate that regret minimization yields correlated
equilibria when the strategy sets are convex and compact.} These results
suggests that correlated equilibrium may be a more natural solution concept
for implementation problems, as it emerges from straightforward adaptive
behavior rather than requiring sophisticated strategic reasoning.\footnote{%
\cite{halpern2016minimizing} and \cite{srivastava2019decision} provide
empirical evidence in support of the claim that decision makers perform
regret minimization.}

While we know which social choice functions can be implemented in Nash
equilibrium, we lack a characterization of SCFs which can be implemented in
correlated equilibria and implementing mechanisms that explicitly target
implementation in correlated equilibria. Such mechanisms would be especially
valuable as they could achieve desired social outcomes under simple learning
dynamics rather than requiring agents to compute equilibria.

This paper provides a step forward in resolving this tension between
implementation theory and behavioral realism. We show that
Maskin-monotonicity characterizes SCFs which are implementable in correlated
equilibrium and provide a mechanism which can fully implement such SCFs as
the outcome obtained in a unique correlated equilibrium. While there is a
voluminous literature on full implementation, our paper offers the first
mechanism that implements in adaptive heuristics per \cite{HartMasColell2000}
. Previous papers such as \cite{cabrales1999adaptive} appeal to both integer
games and adaptive dynamics where every best response must be played with
positive probability, while \cite{cabrales2011implementation} requires that
the agents switch with positive probability to any better or best response.
In addition, we demonstrate in a simulation that the mechanism achieves the
desirable social choice outcome in a canonical bilateral trade environment
when the buyer and seller play the induced game using regret minimization
heuristics. Finally, the mechanism does not rely on integer or modulo games,
making it a reasonable candidate for practical deployment.\footnote{%
See \cite{Jackson92} and \cite{Moore92} for a critique of these tail-chasing
devices.}

\section{Mechanism and Proof of Implementation}

\subsection{Environment}

Consider a finite set of agents $\mathcal{I}=\{1,2,...,I\}$ with $I\geq 2$;
a finite set of possible states $\Theta $; and a set of pure alternatives $A$%
. We consider an environment with lotteries and transfers. Specifically, we
work with the space of allocations/outcomes $X\equiv \Delta \left( A\right)
\times 
%TCIMACRO{\U{211d} }%
%BeginExpansion
\mathbb{R}
%EndExpansion
^{I}$ where $\Delta (A)$ denotes the set of lotteries on $A$ that have a
countable support, and $%
%TCIMACRO{\U{211d} }%
%BeginExpansion
\mathbb{R}
%EndExpansion
^{I}$ denotes the set of transfers to the agents.

Each state $\theta \in \Theta $ induces a type $\theta _{i}\in \Theta _{i}$
for each agent $i\in \mathcal{I}$. Assume that $\Theta $ has no redundancy,
i.e., whenever $\theta \neq \theta ^{\prime }$, we must have $\theta
_{i}\neq \theta _{i}^{\prime }$ for some agent $i$. Hence, we can identify a
state $\theta $ with its induced type profile $\left( \theta _{i}\right)
_{i\in \mathcal{I}}$ and $\Theta $ with a subset of $\times _{i=1}^{I}\Theta
_{i}$. Moreover, we say that a type profile $\left( \theta _{i}\right)
_{i\in \mathcal{I}}$ identifies a state $\theta ^{\prime }$ if $\theta
_{i}=\theta _{i}^{\prime }$ for every $i\in \mathcal{I}$. Each type $\theta
_{i}\in \Theta _{i}$ induces a utility function $u_{i}\left( \cdot ,\theta
_{i}\right) :X\rightarrow \mathbb{R}$ which is quasilinear in transfers and
has a bounded expected utility representation on $\Delta \left( A\right) $.
That is, for each $x=\left( \ell ,\left( t_{i}\right) _{i\in \mathcal{I}%
}\right) \in X$, we have $u_{i}\left( x,\theta _{i}\right) =v_{i}(\ell
,\theta _{i})+t_{i}$ for some bounded expected utility function $v_{i}(\cdot
,\theta _{i})$ over $\Delta \left( A\right) $. That is, we work with an
environment with transferable utility (TU) restriction on agents'
preferences which is absent in \cite{maskin99}. As in \cite{AM92}, we will
take for granted that distinct elements of $\Theta _{i}$ induce different
preference orderings over $\Delta \left( A\right) ,$ and also that a player
is never indifferent over all elements of $A.$

We focus on a \textit{complete information} environment in which the state $%
\theta $ is common knowledge among the agents but unknown to a mechanism
designer. Thanks to the complete-information assumption, it is indeed
without loss of generality to assume that agents' values are private.

The designer's objective is specified by a \textit{social choice function} $%
f:\Theta \rightarrow X$, namely, if the state is $\theta $, the designer
would like to implement the social outcome $f\left( \theta \right) $. We
allow an SCF to be defined as a mapping from $\Theta $ to $X$ only so as to
keep its consistency with the range of the outcome function used in the
implementing mechanism. We can define $f:\Theta \rightarrow \Delta \left(
A\right) $ as a special case of SCFs, as long as the designer is still
allowed to impose off-the-equilibrium transfers in the implementing
mechanism.

\subsection{Implementation}

We denote a mechanism by $\mathcal{M}=\left( (M_{i},\tau _{i}\right) _{i\in 
\mathcal{I}},g)$ where $M_{i}$ is a nonempty \textit{set of messages}
available to agent $i$; $g:M\rightarrow X$ (where $M\equiv \times
_{i=1}^{I}M_{i}$) is the \textit{outcome function}; and $\tau
_{i}:M\rightarrow \mathbb{R}$ is the \textit{transfer rule} which specifies
the payment to agent $i$. At each state $\theta \in \Theta $, the
environment and the mechanism together constitute a \textit{game with
complete information} which we denote by $\Gamma (\mathcal{M},\theta )$.

Let $\sigma \in \Delta (M)$ be a probability distribution over $M$. A
strategy profile $\sigma $ is said to be a correlated \textit{equilibrium}
of the game $\Gamma (\mathcal{M},\theta )$ if, for all agents $i\in \mathcal{%
I}$ and all messages $m_{i}$ such that $(m_{i},m_{-i})\in $supp$\left(
\sigma \right) $ and $m_{i}^{\prime }\in M_{i}$, we have 
\begin{eqnarray*}
&&\sum_{m_{-i}\in M_{-i}}\sigma _{-i}(m_{-i}|m_{i})\left[ \tilde{u}%
_{i}(g(m_{i},m_{-i}),\theta )+\tau _{i}(m_{i},m_{-i})\right] \\
&\geq &\sum_{m_{-i}\in M_{-i}}\sigma _{-i}(m_{-i}|m_{i})\left[ 
\tilde{u}_{i}(g(m_{i}^{\prime },m_{-i}),\theta )+\tau _{i}(m_{i}^{\prime
},m_{-i})\right] \text{,}
\end{eqnarray*}%
where we use $\sigma _{-i}(m_{-i}|m_{i})$ to denote the marginal probability
on $m_{-i}$ conditional on $m_{i}$. For simplicity, we abuse notations to
write $\sigma _{j}(m_{j}|m_{i})$ for the marginal probability on $m_{j}$
conditional on $m_{i}$. \newline
Note that it is equivalent to have the following definition, 
\begin{eqnarray*}
&&\sum_{m_{-i}\in M_{-i}}\sigma (m_{i},m_{-i})\left[ \tilde{u}%
_{i}(g(m_{i},m_{-i}),\theta )+\tau _{i}(m_{i},m_{-i})\right] \\
&\geq &\sum_{m_{-i}\in M_{-i}}\sigma (m_{i},m_{-i})\left[ \tilde{u}%
_{i}(g(m_{i}^{\prime },m_{-i}),\theta )+\tau _{i}(m_{i}^{\prime },m_{-i})%
\right] \text{.}
\end{eqnarray*}%
Let $CE(\Gamma (\mathcal{M},\theta ))$ denote the set of correlated
equilibria of the game $\Gamma (\mathcal{M},\theta )$. We also denote by $%
\mbox{supp}\ (CE(\Gamma (\mathcal{M},\theta )))\ $the set of message
profiles that can be played with positive probability under some correlated
equilibrium $\sigma \in CE(\Gamma (\mathcal{M},\theta )$, i.e., 
\begin{equation*}
\mbox{supp}\ (CE(\Gamma (\mathcal{M},\theta )))=\{m\in M:%
\mbox{there exists
$\sigma \in CE(\Gamma(\mathcal{M}, \theta))$ such that $\sigma(m) > 0$}\}. 
\end{equation*}%
We now define our notion of implementation.

\begin{definition}
\label{def-nash}An SCF $f$ is \textbf{implementable in correlated equilibria}
if there exists a mechanism $\mathcal{M}=\left( (M_{i},\tau _{i}\right)
_{i\in \mathcal{I}},g)$ such that for every state $\theta \in \Theta $, $%
m\in \mbox{supp}\ (CE(\Gamma (\mathcal{M},\theta )))\Rightarrow g\left(
m\right) =f\left( \theta \right) $ and $\tau _{i}\left( m\right) =0$ for
every $i\in \mathcal{I}$.
\end{definition}

\subsection{Maskin Monotonicity}

\label{sec:mono}

For $\left( x,\theta _{i}\right) \in X\times \Theta _{i}$, we use $\mathcal{L%
}_{i}\left( x,\theta _{i}\right) $ to denote the lower-contour set at
allocation $x$ in $X$ for type $\theta _{i}$, i.e.,%
\begin{equation*}
\mathcal{L}_{i}\left( x,\theta _{i}\right) =\left\{ x^{\prime }\in
X:u_{i}\left( x,\theta _{i}\right) \geq u_{i}(x^{\prime },\theta
_{i})\right\} . 
\end{equation*}%
We use $\mathcal{SU}_{i}\left( x,\theta _{i}\right) $ to denote the strict
upper-contour set of $x\in X$ for type $\theta _{i}$, i.e.,%
\begin{equation*}
\mathcal{SU}_{i}\left( x,\theta _{i}\right) =\left\{ x^{\prime }\in
X:u_{i}(x^{\prime },\theta _{i})>u_{i}(x,\theta _{i})\right\} \text{.} 
\end{equation*}%
We now state the definition of Maskin monotonicity which \cite{maskin99}
proposes for Nash implementation.

\begin{definition}
\label{mm}An SCF $f$ satisfies \textbf{Maskin monotonicity} if, for every
pair of states $\tilde{\theta}$ and $\theta $ with $f(\tilde{\theta}%
)\not=f\left( \theta \right) $, there is some agent $i\in \mathcal{I}$ such
that 
\begin{equation}
\mathcal{L}_{i}(f(\tilde{\theta}),\tilde{\theta}_{i})\cap \mathcal{SU}_{i}(f(%
\tilde{\theta}),\theta _{i})\not=\varnothing \text{.}  \label{maskin-m}
\end{equation}
\end{definition}

\noindent The agent $i$ in Definition \ref{mm} is called a \textquotedblleft
whistle-blower" or a \textquotedblleft test agent"; likewise, an allocation
in $\mathcal{L}_{i}(f(\tilde{\theta}),\tilde{\theta}_{i})\cap \mathcal{SU}%
_{i}(f(\tilde{\theta}),\theta _{i})$ is called a \textquotedblleft test
allocation\textquotedblright\ for agent $i$ and the ordered pair of states $(%
\tilde{\theta},\theta )$.

\subsection{\label{M}The Mechanism}

In this section, we present our main result which shows that Maskin
monotonicity is necessary and sufficient for implementation in correlated
equilibrium. We formally state the result as follows:

\begin{theorem}
\label{CE}An SCF $f$ is implementable in \textbf{correlated equilibria} if
and only if it satisfies Maskin monotonicity.
\end{theorem}

\cite{maskin99} establishes that Maskin monotonicity is necessary for
implementation in Nash equilibria. Since each Nash equilibrium is also a
correlated equilibrium, Maskin monotonicity is also necessary for
implementation in correlated equilibria. To show that it is sufficient, we
will construct an implementing mechanism in the remainder of the section
having first established some preliminaries.

\subsubsection{Dictator Lotteries}

\begin{lemma}
\label{AM}For each agent $i\in \mathcal{I}$, there exists a function $%
y_{i}:\Theta _{i}\rightarrow X$ such that for every types $\theta _{i}$ and $%
\theta _{i}^{\prime }$ with $\theta _{i}\neq \theta _{i}^{\prime }$, we have%
\begin{equation}
u_{i}\left( y_{i}\left( \theta _{i}\right) ,\theta _{i}\right) >u_{i}\left(
y_{i}\left( \theta _{i}^{\prime }\right) ,\theta _{i}\right) \text{.}
\label{d}
\end{equation}
\end{lemma}

\cite{AM92} prove the existence of lotteries $\left\{ y_{i}\left( \cdot
\right) \right\} \subset \Delta \left( A\right) $ which satisfy Condition (%
\ref{d}). % To satisfy Condition (\ref{d-w}), we
% simply add a reward of $\eta ^{\prime }$ to each outcome of the lotteries $%
% \left\{ y_{i}^{\prime }\left( \theta _{i}\right) \right\} _{\theta _{i}\in
% \Theta _{i}}$. More precisely, for each $\theta _{i}\in \Theta _{i}$, we set 
% \begin{equation*}
% y_{i}(\theta _{i})=(y_{i}^{\prime }(\theta _{i}),\eta ^{\prime },\ldots
% ,\eta ^{\prime })\in X
% \end{equation*}%
% We call the resulting lotteries the \textit{dictator lotteries} for agent $i$
% and denote them by $\left\{ y_{i}\left( \cdot \right) \right\} $. ??Condition (\ref{d-w}) gives agents strict incentives to choose the outcomes rather the social outcomes in dictator lotteries whenever possible. This helps us to correct behaviour off equilibrium path, and by a proper design of the trigger scheme a unilateral deviation on equilibrium path is not profitable.??
We now define a notion called \emph{the best challenge scheme}, which plays
a crucial role in proving Theorem \ref{CE}. First, a \emph{challenge scheme}
for an SCF $f$ is a collection of (pre-assigned) test allocations $\{x(%
\tilde{\theta},\theta _{i})\}$, one for each pair of state $\tilde{\theta}$
and type $\theta _{i}$ of agent $i$, such that 
\begin{eqnarray}
\text{ if }\mathcal{L}_{i}(f(\tilde{\theta}),\tilde{\theta}_{i})\cap 
\mathcal{SU}_{i}(f(\tilde{\theta}),\theta _{i}) &\neq &\varnothing \text{,
then }x(\tilde{\theta},\theta _{i})\in \mathcal{L}_{i}(f(\tilde{\theta}),%
\tilde{\theta}_{i})\cap \mathcal{SU}_{i}(f(\tilde{\theta}),\theta _{i})\text{%
;}  \label{reverse} \\
\text{if }\mathcal{L}_{i}(f(\tilde{\theta}),\tilde{\theta}_{i})\cap \mathcal{%
SU}_{i}(f(\tilde{\theta}),\theta _{i}) &=&\varnothing \text{, then }x(\tilde{%
\theta},\theta _{i})=f(\tilde{\theta})\text{.}
\end{eqnarray}

When $\mathcal{L}_{i}(f(\tilde{\theta}),\tilde{\theta}_{i})\cap \mathcal{SU}%
_{i}(f(\tilde{\theta}),\theta _{i})\neq \varnothing $, we may think of state 
$\tilde{\theta}$ as an announcement made by another agent(s) which agent $i$
could challenge (as a whistle-blower) when agent $i$ has true type $\theta
_{i}$. The following lemma shows that there is a challenge scheme in which
each whistle-blower $i$ facing state announcement $\tilde{\theta}$ has a
weak incentive to report his true type $\theta _{i}$ to challenge $\tilde{%
\theta}$.

\begin{lemma}
\label{BC}For any SCF $f$, there is a challenge scheme $\{x(\tilde{\theta}%
,\theta _{i})\}_{i\in \mathcal{I},\tilde{\theta}\in \Theta ,\theta _{i}\in
\Theta _{i}}$ such that for every $i\in \mathcal{I},\ \tilde{\theta}\in
\Theta ,$ and $\theta _{i}\in \Theta _{i}$, 
\begin{equation}
u_{i}(x(\tilde{\theta},\theta _{i}),\theta _{i})\geq u_{i}(x(\tilde{\theta}%
,\theta _{i}^{\prime }),\theta _{i}),\forall \theta _{i}^{\prime }\in \Theta
_{i}\text{;}  \label{best-C}
\end{equation}

% moreover, whenever, $x(\tilde{\theta},\theta _{i})\not=f(\tilde{\theta})$,
% we have 
% \begin{equation}
% u_{i}(x(\tilde{\theta},\theta _{i}),\theta _{i}^{\prime \prime
% })\not=u_{i}(z^{\prime },\theta _{i}^{\prime \prime }),\forall \theta
% _{i}^{\prime \prime }\in \Theta _{i},\forall z^{\prime }\in F(\Theta )\text{.%
% % }  \label{gC}
% \end{equation}
\end{lemma}

Having formulated the best challenge scheme, we will detail the mechanism in
the following sections.

\subsubsection{\label{ar}Message Space and Outcome Function}

For each agent $i$, we define the message space as follows. A generic
message 
\begin{equation*}
m_{i}=(m_{i}^{1},m_{i}^{2})\in M_{i}^{1}\times M_{i}^{2}, 
\end{equation*}%
where $M_{i}^{1}=M_{i}^{2}=\times _{j=1}^{I}\Theta _{j}$ is the state space out of which each agent reports messages twice.

For each message profile $m\in M$, the allocation is determined as follows:
\begin{equation*}
g\left( m\right) =\frac{1}{I^{2}}\sum_{i\in \mathcal{I}}\sum_{j\in \mathcal{I%
}}g_{i,j}(m)\text{,} 
\end{equation*}%
and, 
\begin{equation*}
g_{i,j}(m)=\left[ e_{i,j}\left( m_{i},m_{j}\right) \left( \frac{1}{I}%
\sum_{k\in \mathcal{I}}y_{k}(m_{k,k}^{1})\right) \oplus \left(
1-e_{i,j}\left( m_{i},m_{j}\right) \right) x(m_{i}^{2},m_{j,j}^{1})\right] 
\end{equation*}%
\newline
where $\{y_{k}(\cdot )\}\ $are the dictator lotteries for agent $k$ obtained
from Lemma \ref{AM}, moreover, we define 
\begin{equation*}
e_{i,j}(m_{i},m_{j})=\left\{ 
\begin{array}{ll}
0\text{,} & \text{if }m_{i}^{2}\in \Theta \text{, and }%
x(m_{i}^{2},m_{j,j}^{1})=f(m_{i}^{2})\text{;} \\ 
\varepsilon \text{,} & \text{if }m_{i}^{2}\in \Theta \text{, and }%
x(m_{i}^{2},m_{j,j}^{1})\not=f(m_{i}^{2})\text{;} \\ 
1\text{,} & \text{if }m_{i}^{2}\notin \Theta \text{.}%
\end{array}%
\right. 
\end{equation*}%
Hereafter we say % that the second reports of
% agent $i$ and agent $j$ are \emph{consistent} if $f(m_{i}^{2})=f(m_{j}^{2})$, 
% \emph{and} the common type profile identifies a state in $\Theta $;
agent $i$ \emph{is} \emph{challenged} by agent $j$ if $%
x(m_{i}^{2},m_{j,j}^{1})\not=f(m_{i}^{2})$, which is equivalent to $%
g_{i,j}\not=f(m_{i}^{2})$. 
% , and agent $i$ challenges himself if $x(m_{i}^{2},m_{i}^{1})\not =f(m_{i}^{2})$

In words, the designer first chooses an ordered pair(including agent $i$
paired with himself) of agents $(i,j)$ with equal probability. The outcome
function distinguishes three cases: (1) if agent $i$ reports a type profile
which identifies a state in $\Theta ,$ and agent $j$ does not challenge
agent $i$, then we implement 
\begin{equation*}
g_{i,j}(m)=f\left( m_{i}^{2}\right) ; 
\end{equation*}

(2) % if agent $i$ reports a type profile which
% does not identify a state in $\Theta $, then we implement the dictator
% lottery $\frac{1}{I}\sum_{k\in\mathcal{I}}y_{k}\left( m_{k}^{1}\right) $;
% % \footnote{%
% % Observe that we make the first report of both agents $i$ and $j$ effective
% % (through affecting the compound lottery $\frac{1}{2}\sum_{k=i,j}y_{k}\left(
% % m_{k}^{1}\right) $), regardless of whether pair $\left( i,j\right) $ or pair 
% % $\left( j,i\right) $ is picked. The construction will be used in proving
% % Claim \ref{c0}, which, in turn, is used to prove Claim \ref{pc2}.} 
% (3)
if agent $i$ reports a type profile which identifies a state in $\Theta ,$
and agent $j$ does challenge agent $i$, then we implement the compound
lottery\footnote{%
More precisely, if $x=\left( \ell ,\left( t_{i}\right) _{i\in \mathcal{I}%
}\right) $ and $x^{\prime }=\left( \ell ^{\prime },\left( t_{i}^{\prime
}\right) _{i\in \mathcal{I}}\right) $ are two outcomes in $X$, we identify $%
\alpha x\oplus \left( 1-\alpha \right) x^{\prime }$ with the outcome $\left(
\alpha \ell \oplus \left( 1-\alpha \right) \ell ^{\prime },\left( \alpha
t_{i}+\left( 1-\alpha \right) t_{i}^{\prime }\right) _{i\in \mathcal{I}%
}\right) $. For simplicity, we also write the compound lottery 
% $\frac{1}{2}%
% y_{i}\left( m_{i}^{1}\right) \oplus \frac{1}{2}y_{j}\left( m_{j}^{1}\right) $
as $\frac{1}{I}\sum_{k\in \mathcal{I}}y_{k}(m_{k}^{1})$.}%
\begin{equation*}
g_{i,j}(m)=C_{i,j}^{\varepsilon }(m_{i},m_{j})\equiv \varepsilon \left( 
\frac{1}{I}\sum_{k\in \mathcal{I}}y_{k}\left( m_{k,k}^{1}\right) \right)
\oplus \left( 1-\varepsilon \right) x(m_{i}^{2},m_{j,j}^{1})\text{.} 
\end{equation*}%
Note that $C_{i,j}^{\varepsilon }(m_{i},m_{j})$ is an $\left( \varepsilon
,1-\varepsilon \right) $-combination of (i) the dictator lotteries which
occur with equal probability; and (ii) the allocation specified by the best
challenge scheme $x(m_{i}^{2},m_{j,j}^{1})$;

(3) if agent $i$ reports a type profile which does not identify a state in $%
\Theta $, then we implement the dictator lottery $\frac{1}{I}\sum_{k\in 
\mathcal{I}}y_{k}\left( m_{k,k}^{1}\right) .$

We abuse notations to write $x(m_{i}^{2},m_{j,j}^{1})\ $and $f(m_{i}^{2})$
for all $i$ and $j$ assuming $m_{i}^{2}$ identifies a state since the
outcome will be only determined by the dictator lottery otherwise.

Given the construction of dictator lotteries, we can choose $\varepsilon >0$
sufficiently small, and $\eta >0$ sufficiently large % \footnote{%
% Instead of using $\eta ^{\prime }$ defined in (\ref{D'}), we choose $\eta $
% because the mechanism may produce a strictly larger finite set of
% alternatives than those contained in $\tilde{X}$. For instance, allocations
% from the dictator lotteries may occur from the mechanism but are not
% contained in $\tilde{X}$.} 
such that firstly, we have 
\begin{equation}
\eta >\sup_{i\in \mathcal{I},\theta _{i}\in \Theta _{i},m,m^{\prime }\in
M}\left\vert u_{i}(g\left( m\right) ,\theta _{i})-u_{i}(g\left( m^{\prime
}\right) ,\theta _{i})\right\vert \text{;}  \label{D}
\end{equation}%
secondly, it does not disturb the \textquotedblleft
effectiveness\textquotedblright\ of agent $j$'s challenge: due to (\ref%
{reverse}), we can have%
\begin{eqnarray}
x(m_{i}^{2},m_{j,j}^{1}) &\neq &f(m_{i}^{2})\Rightarrow  \notag \\
u_{j}(C_{i,j}^{\varepsilon }(m_{i},m_{j}),m_{i,j}^{2})
&<&u_{j}(f(m_{i}^{2}),m_{i,j}^{2})\text{ and }u_{j}(C_{i,j}^{\varepsilon
}(m_{i},m_{j}),m_{j,j}^{1})>u_{j}(f(m_{i}^{2}),m_{j,j}^{1})\text{.}
\label{bw}
\end{eqnarray}%
%
%
%
%
%
%
%
%
%
%
%
%
%
%
%
%
%
%
%
%
%
%
%
%
%
%
%
%
%
% thirdly, by choosing $\varepsilon_k$ properly for each $k$, we can make sure that for every $i\neq j$, 
% \begin{equation}
%     C_{i,i}^{\varepsilon_i }(m_i,m_i)\neq C_{j,j}^{\varepsilon_j }(m_j,m_j)
% \end{equation},
% and for every $\theta_i, \theta^{\prime}$ 
% \begin{equation}
%     u_i(C_{i,i}^{\varepsilon_i }(m_i,m_i),\theta_i)\neq u_i(f(\theta^{\prime}),\theta_i).
% \end{equation}
It means that whenever agent $j$ challenges agent $i$, the lottery $%
C_{i,j}^{\varepsilon }(m_{i},m_{j})$ is strictly worse than $f\left(
m_{i}^{2}\right) $ for agent $j$ when agent $i$ tells the truth about agent $%
j$'s preference in $m_{i}^{2}$; moreover, the lottery $C_{i,j}^{\varepsilon
}(m_{i},m_{j})$ is strictly better than $f\left( m_{i}^{2}\right) $ for
agent $j$ when agent $j$ tells the truth in $m_{j,j}^{1}$, which implies
that agent $i$ tells a lie about agent $j$'s preference. 
% In addition by Lemma \ref{BCC}, we can make sure the genericity property holds as stated in the third point above.

\subsubsection{\label{Fmoney}Transfer Rule}

We now define the transfer rule. For every message profile $m\in M$ and
every agent $i\in \mathcal{I}$, we specify the transfer received by agent $i$
as follows: 
\begin{equation*}
\tau _{i}(m)=\sum_{j\neq i}\left[ \tau
_{i,j}^{1}(m_{i},m_{j})+\tau _{i,j}^{2}(m_{i},m_{j})+\tau
_{i,j}^{3}(m_{i},m_{j})\right]. 
\end{equation*}

% For the transfer $sc_{i,j}$, we define it as a strict scoring rule, a
% \textquotedblleft bet\textquotedblright\ on his opponent agent $j$
% challenging agent $j$ himself. Thus, the transfers depend on the probability
% reported by agent $i$ on the event agent $j$ being challenged, $c_{i,j}$ and
% the one realized. We define the proper scoring rule explicitly as follows, 
% \begin{equation*}
% sc_{i,j}\left( m\right) =\left\{ 
% \begin{array}{ll}
% -(c_{i,j})^{2}-(1-c_{i,j})^{2}+2c_{i,j}-1\text{, } & \text{if }%
% x(m_{j}^{2},m_{j,j}^{1})\neq f(m_{j}^{2}); \\ 
% -(c_{i,j})^{2}-(1-c_{i,j})^{2}+2(1-c_{i,j})-1\text{,} & \text{otherwise.}%
% \end{array}%
% \right. 
% \end{equation*}%
%
%
%
%
%
%
%
%
%
%
%
%
%
%
%
%
%
%
%
%
%
%
%
%
%
%
%
%
%
% By design, it is easy to see that with belief $\sigma_j(B_{m^2_i,j}|m_i)$ for agent $i$, it is strictly better for agent $i$ to report $c^1_{i,j}=\sigma_j(B_{m^2_i,j}|m_i)$ rather than any other belief. 
% \begin{comment}
%     When $\sigma_j(B_{m^2_i,j}|m_i)=0$, we have $sc^1_{i,j}=0.$
% \end{comment}
\begin{eqnarray}
\tau _{i,j}^{1}\left( m_{i},m_{j}\right)  &=&\left\{ 
\begin{tabular}{ll}
$-2\eta $ , & \text{if} $x(m_{i}^{2},m_{j,j}^{1})\neq f(m_{i}^{2})$; \\ 
$0$, & \text{otherwise.}%
\end{tabular}%
\ \right.   \label{bcha} \\
\tau _{i,j}^{2}\left( m_{i},m_{j}\right)  &=&\left\{ 
\begin{tabular}{ll}
$-\epsilon $ , & \text{if} $x(m_{j}^{2},m_{j,j}^{1})\neq f(m_{j}^{2})$ \text{
and } $m_{i,j}^{1}\neq m_{j,j}^{1}$; \\ 
$-\epsilon $, & if $x(m_{j}^{2},m_{j,j}^{1})=f(m_{j}^{2})$  and  $%
x(m_{j}^{2},m_{i,j}^{1})\neq f(m_{j}^{2})$; \\ 
$0$, & \text{otherwise.}%
\end{tabular}%
\ \right.   \label{spotother} \\
\tau _{i,j}^{3}\left( m_{i},m_{j}\right)  &=&\left\{ 
\begin{array}{ll}
-\eta \text{,} & \text{if }x(m_{i}^{2},m_{j,i}^{1})\neq f(m_{i}^{2}){;} \\ 
0\text{,} & \text{otherwise.}%
\end{array}%
\right.   \label{backtotruth}
\end{eqnarray}%
Recall that $\eta >0$ is chosen to be greater than the maximal utility
difference from the outcome function $g\left( \cdot \right) $; see (\ref{D}%
), and $\epsilon $ is an arbitrarily small positive number.

% In words, for  $\tau ^{1},$, given each pair of agents $(i,j)$, agent $i$'S transfers depends on his bet on whether agent $j$ is challenged. Specifically, it is â ; if the allocations induced by their second reports differ ($(g_{i,i}(m_{i})\not=g_{j,j}(m_{j}))$), then we consider the following two
% subcases: (i) if agent $i$'s second report about agent $j$'s type matches
% agent $j$'s first report ($m_{i,j}^{2}=m_{j}^{1}$), then agent $j$ pays $%
% \eta $ to agent $i$; (ii) if agent $i$'s second report about agent $j$'s
% type does not match agent $j$'s first report ($m_{i,j}^{2}\neq m_{j}^{1}$),
% then both agents pay $\eta $ to the designer. Note that the first report $%
% m_{i}^{1}$ does not affect the transfer to agent $i$.
Given each agent $j\neq i$, we define the first transfer $\tau ^{1}$ for
agent $i$ so that agent $i$ is asked to pay $2\eta $ if agent $i$ is
challenged by agent $j$ or agent $i$ reports a type profile which does not
identify a state in $\Theta $.

We will first show that in equilibrium, no agent $i$ reports $m_{i}^{2}$
which does not identify a state in $\Theta$ because it immediately incurs the largest penalty for agent $i$
regarding how to report $m_{i}^{2}.$

For $\tau ^{2},$ it is designed so that agent $i$ is asked to pay $\epsilon $
if agent $j$ challenges himself, and agent $i$'s first report on agent $j$'s
type differs from the reported one from agent $j$ in agent $j$'s first
report. Note that later we will argue that whenever agent $j$ challenges
himself agent $j$ reports the truth in his own type from the first report,
thus, it is a weakly dominant strategy for agent $i$ to report agent $j$'s
true type in agent $i$'s first report.

For $\tau ^{3},$ it is designed such that given agent $j$ makes a positive
bet on agent $i$ challenging himself, agent $i$ is asked to pay $\eta $ if
his second report on his own type is different from the one reported by
agent $j$ in agent $j$'s first report.

\subsection{\label{proof}Proof of Implementation}

% As a consequence of Lemmas \ref{AM} and \ref{BCC}, the mechanism has the
% following crucial property of which we will make use in establishing the
% implementation.

% \begin{claim}
% \label{cc0}Let $\sigma $ be a Nash equilibrium of the game $\Gamma (\mathcal{%
% M},\theta )$. If $m_{i}^{1}\neq \theta _{i}$ for some $m_{i}\in $supp$%
% (\sigma _{i})$, then for each agent $j\not=i$, we have $e_{i,j}\left(
% m_{i},m_{j}\right) =e_{j,i}\left( m_{j},m_{i}\right) =0$ with $\sigma _{j}$%
% -probability one. In addition, for any each agents $j,k\in \mathcal{I}%
% \backslash \left\{ i\right\} $, $m_{j}\in $supp$\left( \sigma _{j}\right) $,
% and $m_{k}\in $supp$\left( \sigma _{k}\right) ,$ we have $e_{j,k}\left(
% m_{j},m_{k}\right) =e_{k,j}\left( m_{k},m_{j}\right) =0$.
% \end{claim}
% The proof is in ???.
% \begin{claim}
% \label{1st}The following two statements hold:

% \noindent (a) If agent $j$ sends a truthful first report with $\sigma _{j}$%
% -probability one, then every agent $i\neq j$ must report agent $j$'s type
% truthfully in his second report with $\sigma _{i}$-probability one.

% \noindent (b) If every agent $i\neq j$ reports the same type $\tilde{\theta}%
% _{j}$ of agent $j$ in his second report with $\sigma _{i}$-probability one,
% then agent $j$ must also report the type $\tilde{\theta}_{j}$ in his second
% report with $\sigma _{j}$-probability one.
% \end{claim}
% The proof is in ???.
% \bigskip
We outline our proof strategy as follows. We will first show that in
equilibrium, no agent $i$ reports $m_{i}^{2}$ which does not identify a
state in $\Theta .$ Then, we prove that every agent reports a state which
induces the socially desired outcome and no one challenges anyone, thus,
there is no transfer incurred on equilibrium path.

\begin{enumerate}
\item First, we show that in equilibrium, if an agent $i$ reports a state in
his second report, inducing an outcome different from the socially desired
one (we call it a false state), then agent $i$ must be challenged (maybe by
himself). This is established by monotonicity.

\item Second, in equilibrium, no one reports a false state challenged by
another agent. This is obtained by the first step and the penalty imposed on
agent if he is challenged.

\item Third, we show that no one reports a false state challenged by
himself. This is the key step and difficult point in the proof.

\begin{enumerate}
\item We first show that if agent $j$ has a positive belief that agent $i$
challenges himself, then agent $j$ must report agent $i$'s true type in agent $j$'s first report
(Note that this is due to $\tau^2$ and the scale of $\epsilon$ can be
arbitrarily small since agent $j$'s first report about agent $i$'s type only
affect agent $i$ through $\tau_i^2$).

\item Second, we show that from the agent $i$'s perspective, if he
challenges himself, then all the messages from his opponent report the agent $i$'s true type in
their first reports, thus, agent $i$ suffers $\eta$ with probability $1$ according to $\tau_i^3$.

\item The third transfer then provides a strict incentive for agent $i$ to
report the truth in his second report instead of reporting a lie being
challenged by himself.
\end{enumerate}

We then show that in the equilibrium path, the agent $i$ does not challenge
himself.
\end{enumerate}

% Third, by the first step, in equilibrium, for agent $i$ being challenged, conditional on that particular message of agent $i$, all the message profile reported by agent $i$'s opponents must bet that agent $i$ is challenged. In step two, all agents other than agent $i$, they bet and report their own true types and agent $i$'s true type in their first reports. Finally, we show that agent $i$ should report the true state in his second report, hence contradiction is obtained. 
% under true state statean arbitrary agent $i$ reports a state such that agent $i$ is the whistle blower for that state outcome, then agent $i$ challenges himself (effectively). Second, we show that in equilibrium, if an arbitrary agent $i$ reports a state under which the social outcome is the true state one, and it is not challenged by himself, then the reported state by agent $i$ is challenged by some agent $j$. Third, in equilibrium, no agent will report a state under which the social outcome is the true state one, and it is not challenged by himself. Fourth, in equilibrium, if an arbitrary agent $i$ reports a state such that agent $i$ is the whistle blower for that state outcome, then for every message with positive equilibrium support, agent $i$ challenges himself. Fifth, in equilibrium, there is no self-challenge. Sixth, finally, no transfers, no challenge. 
\bigskip

% \begin{claim}
% \label{state}Given an arbitrary agent $i$ and message profile $m\in $supp$%
% \left( \sigma \right) $, we have $m_{i}^{2}\in \Theta .$
% \end{claim}

% \begin{proof}
% Suppose not. Then from agent $i$'s perspective, conditional on playing $%
% m_{i},$ agent $i$ is penalized $2\eta $ due to (\ref{bcha})), and dictator
% lotteries are triggered with positive 1. By Lemmas \ref{AM} and \ref{BC}, we
% know that for every $\tilde{m}_{-i}$ such that $\sigma (m_{i},\tilde{m}%
% _{-i})>0$, for every agent $j\neq i$, we have $\tilde{m}_{j,j}^{1}=\theta
% _{j}$. Note that the allocation difference is less than $\eta ,$ and all the
% transfers resulting from $m_{i}^{2}$ is on $\tau _{i,j}^{1}\ $and $\tau
% _{i,j}^{3}.$ Considering a deviation to true profiles $\left( \theta
% _{j}\right) _{j\in \mathcal{I}}$ rather than $m_{i}^{2},$ for $\tau
% _{i,j}^{1},$ agent $i$ avoids paying the penalty $2\eta ,$ while regardless
% of whether $c_{j,i}$ is positive the possible loss is bounded by $\eta .$
% Therefore, it is a profitable deviation. This contradicts to that $\sigma $
% is an equilibrium.
% \end{proof}

Now, for each state $\tilde{\theta}\in \Theta$, we define the following set
of agents: 
\begin{equation*}
\mathcal{J}( \tilde{\theta}) \equiv \left\{ j\in \mathcal{I}:%
\mathcal{L}_{j}(f(\tilde{\theta}),\tilde{\theta}_{j})\cap \mathcal{SU}_{j}(f(%
\tilde{\theta}),\theta _{j})=\varnothing \right\} \text{.} 
\end{equation*}

\begin{claim}
\label{whistle}Given an arbitrary agent $i$ and message profile $m\in $supp$%
\left( \sigma \right) $, we have $x(m_{i}^{2},m_{j,j}^{1})\not=f(m_{i}^{2})\ 
$if and only if $j\not\in \mathcal{J}\left( m_{i}^{2}\right) $.
\end{claim}

\begin{proof}
Fix agent $i\in \mathcal{I}$ and a message profile $m\in $supp$(\sigma )$.
We first prove the if-part. Suppose, on the contrary, that $%
x(m_{i}^{2},m_{j,j}^{1})=f(m_{i}^{2})$. Then, consider the deviation from $%
m_{j}$ to $\tilde{m}_{j}$ such that $\tilde{m}_{j,j}^{1}=\theta _{j}$ while $%
m_{j}$ and $\tilde{m}_{j}$ are the same in every other dimension. By (\ref%
{bw}), it delivers a strictly better payoff for agent $j$ against $m_{-j}$
where $j\not\in \mathcal{J}\left( m_{i}^{2}\right) $. At the same time, by
Lemmas \ref{AM} and \ref{BC}, the deviation from $m_{i}$ to $\tilde{m}_{i}$
generates no payoff loss for agent $i$ against any $m_{-i}^{\prime
}\not=m_{-i}$. Thus, the deviation $\tilde{m}_{j}$ is profitable, which
contradicts the hypothesis that $\sigma $ is an equilibrium of the game $%
\Gamma (\mathcal{M},\theta )$.

Next, we prove the only-if-part. Suppose, on the contrary, that there exists
some agent $j\in \mathcal{J}\left( m_{i}^{2}\right) $ such that $%
x(m_{i}^{2},m_{j,j}^{1})\not=f(m_{i}^{2}).$ Since $j\in \mathcal{J}\left(
m_{i}^{2}\right) $, we must have $m_{j,j}^{1}\neq \theta _{j}$. Define $%
\tilde{m}_{j}$ as a deviation which is identical to $m_{j}$ except that $%
\tilde{m}_{j,j}^{1}=\theta _{j}\neq m_{j,j}^{1}$. Then we have $x(m_{i}^{2},%
\tilde{m}_{j,j}^{1})=x(m_{i}^{2},\theta _{j})=m_{i}^{2}$ since $j\in 
\mathcal{J}\left( m_{i}^{2}\right) .$ By (\ref{bw}), $\tilde{m}_{j}$
generates a strictly better payoff for agent $j$ than $m_{j}$ against $%
m_{-j} $. By Lemmas \ref{AM} and \ref{BC}, we also know that agent $j$'s
payoff generated by $\tilde{m}_{j}$ is at least as good as that generated by 
$m_{j}$ against any $m_{-j}^{\prime }\neq m_{-j}$. Hence, $\tilde{m}_{j}$
constitutes a profitable deviation, which contradicts the hypothesis that $%
\sigma $ is an equilibrium of the game $\Gamma (\mathcal{M},\theta )$.
\end{proof}

\begin{claim}
\label{other} No one challenges an allocation announced in the second report
of any other agent, i.e., for any pair of agents $i,j\in \mathcal{I}$ with $%
i\neq j$ and any $m\in $supp$\left( \sigma \right) $, $%
x(m_{i}^{2},m_{j,j}^{1})=f(m_{i}^{2}) $.
\end{claim}

\begin{proof}
Suppose to the contrary that there exist $i,j\in \mathcal{I}$ with $i\neq j$%
, $m\in \mbox{supp}(\sigma )$ such that $x(m_{i}^{2},m_{j,j}^{1})\neq
f(m_{i}^{2})$. By Claim \ref{whistle}, $j\not\in \mathcal{J}\left(
m_{i}^{2}\right) $. By Claim \ref{whistle}, we know that for any agent $%
j\not\in \mathcal{J}\left( m_{i}^{2}\right) ,$ we have $x(f(m_{i}^{2}),\hat{m%
}_{j,j}^{1})\neq f(m_{i}^{2})$, for every $\hat{m}_{-i}$ such that $\sigma
(m_{i},\hat{m}_{-i})>0$. In addition, $\hat{m}_{j}^{1}=\theta _{j}$ for
every $j$ such that $j\not\in \mathcal{J}\left( m_{i}^{2}\right) .$ Hence,
from agent $i$'s perspective, conditional on playing $m_{i},$ he is
challenged with probability $1.$ Thus, agent $i$ is penalized $2\eta $ due
to (\ref{bcha})). Consider a deviation to $\tilde{m}_{i}$ which is the same
as $m_{i}$ except that $\tilde{m}_{i}^{2}=\left( \theta _{j}\right) _{j\in 
\mathcal{I}}$. Note that from agent $i$'s perspective, dictator lotteries
are triggered with probability one, hence, all the agents report the truth
in the first reports. Then, agent $i$ avoids paying the penalty $2\eta $ for
being challenged, while the potential loss from allocation is bounded by $%
\eta ,$ the loss from (\ref{backtotruth}) is bounded by $\eta .$ Therefore,
it is a profitable deviation. This contradicts to that $\sigma $ is an
equilibrium.
\end{proof}

\begin{claim}
\label{bet}For any pair of agents $i$ and $j$ such that $i\neq j$, message
profile $m\in $supp$\left( \sigma \right) $, if $x(m_{j}^{2},m_{j,j}^{1})%
\neq f(m_{j}^{2})$, then we have $m_{i,j}^{1}=\theta _{j}$ and $%
m_{i,i}^{1}=\theta _{i}$.
\end{claim}

\begin{proof}
First, we fix an arbitrary correlated equilibrium $\sigma $. Given $m_{i}$
we collect all the opponents' message profile together with $m_{i}$ agent $i$
knows that agent $j$ is challenged by himself: 
\begin{equation*}
E_{j}(m_{i})=\{\tilde{m}_{-i}:x(\tilde{m}_{j}^{2},\tilde{m}_{j,j}^{1})\neq f(%
\tilde{m}_{j}^{2})\text{ and }\sigma (m_{i},\tilde{m}_{-i})>0.\}
\end{equation*}%
From the condition in Claim \ref{bet}, we know that $E_{j}(m_{i})\neq
\varnothing $. Conditional on $m_{i}$, dictator lotteries are triggered with
positive probability; hence by Lemmas \ref{AM} and \ref{BC}, $%
m_{i,i}^{1}=\theta _{i}$. Note that for any $\tilde{m}_{-i}$ such that $%
\sigma _{i}(\tilde{m}_{-i}|m_{i})>0$ and $\tilde{m}_{-i}\not\in E_{j}(m_{i}),
$ $j$ does not challenge himself. Thus by Claim \ref{whistle}, $x(\tilde{m}%
_{j}^{2},\theta _{j})=f(\tilde{m}_{j}^{2}).$ Hence, according to $\tau
_{i,j}^{2}$ (see (\ref{spotother})), we have $m_{i,j}^{1}=\theta _{j}.$ In
addition, when agent $j$ is challenged, we have $m_{j,j}^{1}=\theta _{j}$.%
% ??DEL?It is easy to see that compared to any $\tilde{m_{i}}$ which is
% identical to $m_{i}$ but $\tilde{m}_{i,j}^{1}\neq \theta _{j}$, $m_{i}$
% provides a better payoff against $m_{-i}$, while is at least as good as $%
% \tilde{m_{i}}$ against any other $\tilde{m}_{-i}$.??
\end{proof}

% \begin{claim}
% \label{betu}For every agent $j$, message profile $m\in $supp$\left( \sigma \right) $, such that $x(m^2_j,m^1_{j,j})\neq f(m^2_j)$, then for every $\tilde{m}_{-j}$ such that $\sigma(m_j,\tilde{m}_{-j})>0$, for every agent $i\neq j$, we have $\tilde{c}_{i,j}>0$, $\tilde{m}^1_{i,j}=\theta_j$ and $\tilde{m}^1_{i,i}=\theta_i$.\end{claim}
% \begin{proof}

% \end{proof}

\begin{claim}
\label{nochallenge}For every agent $j$, message profile $m\in $supp$\left(
\sigma \right) $, we have $x(m^2_j,m^1_{k,k})= f(m^2_j)$ for every agent $k$.
\end{claim}

\begin{proof}
By Claim \ref{other}, we know that for every $k\neq j$, Claim \ref%
{nochallenge} holds. It remains to show that Claim \ref{nochallenge} holds
for $k=j$. Suppose there exists $m\in $supp$\left( \sigma \right) $, such
that $x(m_{j}^{2},m_{j,j}^{1})\not=f(m_{j}^{2})$. By Lemmas \ref{AM} and \ref%
{BC}, $m_{j,j}^{1}=\theta _{j}$. Now we show that for agent $j$ who reports $%
m_{j}$ with $m_{j,j}^{2}\not=\theta _{j}$, it is strictly better for agent $%
j $ to deviate to $\tilde{m}_{j}$ which is identical to $m_{j}$ but $\tilde{m%
}_{j,j}^{2}=\theta $ the true type profile. Specifically, due to Claim \ref%
{bet}, for every $\tilde{m}_{-j}$ such that $\sigma (m_{j},\tilde{m}_{-j})>0$%
, for every agent $i\neq j$, we have $\tilde{m}_{i,j}^{3}>0$, $\tilde{m}%
_{i,j}^{1}=\theta _{j}$ and $\tilde{m}_{i,i}^{1}=\theta _{i}$. By $\tau ^{3}$
in (\ref{backtotruth}), it is strictly better, and for $\tau ^{1}$ and $\tau
^{2},$ there is no loss incured. Hence, we know that the deviation is
profitable. Thus, it contradicts the hypothesis that $%
x(m_{j}^{2},m_{j,j}^{1})\not=f(m_{j}^{2})$.
\end{proof}

\begin{claim}
For every $m\in\text{supp}(\sigma)$, $g(m)=f(\theta)$, and $\tau_{i,j}(m)=0$
for every agent $i$ and $j$.
\end{claim}

\begin{proof}
By Claim \ref{whistle} and Claim \ref{nochallenge}, we know that for every $%
m\in \text{supp}(\sigma )$, $g(m)=f(\theta )$, and $\tau _{i,j}^{1}=0$ for
every agent $i$ and $j$. Due to the construction of the proper scoring rule $%
sc$, we know that $m_{i,j}^{3}=0$ for every agent $i$ and $j$. Hence $\tau
_{i,j}^{2}=\tau _{i,j}^{3}=0$ for every agent $i$ and $j$. Hence, we achieve
implementation with no transfers incurred.
\end{proof}

\begin{remark}
Throughout the proof of our main theorem, to argue the behavior on the
equilibrium path, whenever we use a possible profitable deviation to derive
a contradiction, we use the joint true type profile. Therefore, a feature of
our mechanism is that even with a restricted message space, as long as the
true type profiles are available, our mechanism would still implement.
\end{remark}

\subsection{Social Choice Correspondences}

\label{SCC}

A large portion of the implementation literature strives to deal with social
choice correspondences (hereafter, SCCs), i.e., multi-valued social choice
rules. In this section, we extend our implementation result to cover the
case of SCCs. We suppose that the designer's objective is specified by an
SCC $F:\Theta \rightrightarrows X$; and for simplicity, we assume that $%
F\left( \theta \right) $ is a finite set for each state $\theta \in \Theta $%
. It includes the special case where the co-domain of $F$ is $A$. Following 
\cite{maskin99}, we first define the notion of Nash implementation for an
SCC.

\begin{definition}
\label{Def-implement-SCC}An SCC $F$ is \textbf{implementable in correlated
equilibria} \textbf{by a finite mechanism} if there exists a mechanism $%
\mathcal{M}=\left( (M_{i},\tau _{i})_{i\in \mathcal{I}},g\right) $ such that
for every state $\theta \in \Theta $, the following two conditions are
satisfied: (i) for every $x\in F(\theta )$, there exists a pure-strategy
Nash equilibrium $m\ $in the game $\Gamma (\mathcal{M},\theta )$ with $%
g(m)=x $ and $\tau _{i}(m)=0$ for every agent $i\in \mathcal{I}$; and (ii)
for every $m\in \mbox{supp}\ (CE(\Gamma (\mathcal{M},\theta )))$, we have $%
\mbox{supp}(g(m))\subseteq F\left( \theta \right) $ and $\tau _{i}\left(
m\right) =0$ for every agent $i\in \mathcal{I}$.
\end{definition}

Second, we state the definition of Maskin monotonicity for an SCC.

\begin{definition}
\label{mono-SCC}An SCC $F$ satisfies \textbf{Maskin monotonicity} if for
each pair of states $\tilde{\theta}$ and $\theta $ and $z\in F(\tilde{\theta}%
)\backslash F\left( \theta \right) $, some agent $i\in \mathcal{I}$ and some
allocation $z^{\prime }\in X$ exist such that 
\begin{equation*}
\tilde{u}_{i}(z^{\prime },\tilde{\theta})\leq \tilde{u}_{i}(z,\tilde{\theta})%
\text{ and }\tilde{u}_{i}(z^{\prime },\theta )>\tilde{u}_{i}(z,\theta )\text{%
.} 
\end{equation*}
\end{definition}

We now state our implementation result for SCCs and relegate the proof to
Appendix \ref{NashC}.\footnote{%
When there are only two agents, we can still show that every
Maskin-monotonic SCC $F$ is \textit{weakly} implementable in Nash
equilibria, that is, there exists a mechanism which has a pure-strategy Nash
equilibrium and satisfies requirement (ii) in Definition \ref%
{Def-implement-SCC}.}

\begin{theorem}
\label{NashC}Suppose there are at least three agents. An SCC $F$ is
implementable in correlated equilibria by a finite mechanism if and only if
it satisfies Maskin monotonicity.
\end{theorem}

Compared with Theorem \ref{CE} for SCFs, Theorem \ref{NashC} needs to
overcome additional difficulties. In the case of SCCs, each allocation $x\in
F\left( \theta \right) $ has to be implemented as the outcome of some
pure-strategy equilibrium. Hence, each agent must also report an allocation
to be implemented. It also follows that a challenge scheme for an SCC must
be defined for a type $\theta _{i}$ to challenge a pair $(\tilde{\theta},x)$
with $x\in F(\tilde{\theta})$.

\noindent \textbf{Remark.} Discuss NK here.

\subsection{Proof of Theorem \protect\ref{NashC}}

\label{proof: NashC}

We first extend the notion of\emph{\ }a\emph{\ challenge scheme} for an SCC.
Fix agent $i$ of type $\theta _{i}$. For each state $\tilde{\theta}\in
\Theta $ and $z\in F(\tilde{\theta}),$ if $\mathcal{L}_{i}(z,\tilde{\theta}%
_{i})\cap \hspace{0.1cm}\mathcal{SU}_{i}(z,\theta _{i})\neq \varnothing $,
we select some $x(\tilde{\theta},z,\theta _{i})\in \mathcal{L}_{i}(z,\tilde{%
\theta}_{i})\cap \mathcal{SU}_{i}(z,\theta _{i})$; otherwise, we set $x(%
\tilde{\theta},z,\theta _{i})=z$.

As in the case of SCFs, the following lemma shows that there is a challenge
scheme under which truth-telling induces the best allocation.

\begin{lemma}
\label{BCC}For any SCC $F$, there is a challenge scheme $\{x(\tilde{\theta}%
,z,\theta _{i})\}_{i\in \mathcal{I},\tilde{\theta}\in \Theta ,z\in F(\tilde{%
\theta}),\theta _{i}\in \Theta _{i}}$ such that for every $i\in \mathcal{I}%
,\ \tilde{\theta}\in \Theta ,\ z\in F(\tilde{\theta})$, and $\theta _{i}\in
\Theta _{i}$, 
\begin{equation}
u_{i}(x(\tilde{\theta},z,\theta _{i}),\theta _{i})\geq u_{i}(x(\tilde{\theta}%
,z,\theta _{i}^{\prime }),\theta _{i}),\forall \theta _{i}^{\prime }\in
\Theta _{i}\text{;}  \label{best-C}
\end{equation}

% moreover, whenever, $x(\tilde{\theta},z,\theta _{i})\not=z$, we have 
% \begin{equation}
% u_{i}(x(\tilde{\theta},z,\theta _{i}),\theta _{i}^{\prime \prime
% })\not=u_{i}(z^{\prime },\theta _{i}^{\prime \prime }),\forall \theta
% _{i}^{\prime \prime }\in \Theta _{i},\forall z^{\prime }\in F(\Theta )\text{.%
% }  \label{gC}
% \end{equation}
\end{lemma}

Lemma \ref{best-C} is established when we apply the proof of Lemma \ref{BC}
to the challenge scheme $\{x(\tilde{\theta},z,\theta _{i})\}_{i\in \mathcal{I%
},\tilde{\theta}\in \Theta ,z\in F(\tilde{\theta}),\theta _{i}\in \Theta
_{i}}$. Thus, we omit the proof here. % \end{proof}

Next, we propose a mechanism $\mathcal{M}=\left( (M_{i}),g,(\tau
_{i})\right) _{i\in \mathcal{I}}$ which will be used to prove the if-part of
Theorem \ref{NashC}. First, a generic message of agent $i$ is described as
follows: 
\begin{gather*}
m_{i}=\left( m_{i}^{1},m_{i}^{2},m_{i}^{3},m_{i}^{4}\right) \in
M_{i}=M_{i}^{1}\times M_{i}^{2}\times M_{i}^{3}\times M_{i}^{4}=\Theta
\times \Theta \times F(\Theta )\times C_{i}^{1}\text{ s.t.} \\
m_{i}^{3}\in F(m_{i}^{2})\text{. }
\end{gather*}%
That is, agent $i$ is asked to announce (1) agent $i$'s own type (which we
denote by $m_{i}^{1}$); (2) a type profile (which we denote by $m_{i}^{2}$);
(3) an allocation $m_{i}^{3}$ such that $m_{i}^{3}\in F(m_{i}^{2})$ if $%
m_{i}^{2}$ is a state;(4) $C_{i}^{1}=\times _{j\neq i}C_{i,j}$ with $%
C_{i,j}=[0,1]$, where each agent reports the information about his
opponents' status on \textquotedblleft being challenged\textquotedblright
(defined as follows). As we do in the case of SCFs, we write $m_{i,j}^{2}=%
\tilde{\theta}_{j}$ if agent $i$ reports in $m_{i}^{2}$ that agent $j$'s
type is $\tilde{\theta}_{j}$. % Likewise, since $F$ is Maskin-monotonic, we
% have $F\left( \theta \right) =F\left( \tilde{\theta}\right) $ if $\tilde{%
% \theta}_{i}=\theta _{i}$ for every $i$; hence, for $m_{i}^{2}\in \Theta $, $%
% F(m_{i}^{2})$ is uniquely defined as~$F\left( \tilde{\theta}\right) $ such
% that $\tilde{\theta}_{j}=m_{i,j}^{2}$ for all $j$.

We define $\phi (m)$ as follows: for each $m\in M$, 
\begin{equation*}
\phi (m)=\left\{ 
\begin{tabular}{l}
$x$, \\ 
$m_{1}^{3}$,%
\end{tabular}%
\begin{tabular}{l}
if $\left\vert \left\{ i\in \mathcal{I}:m_{i}^{3}=x\right\} \right\vert \geq
I-1$; \\ 
otherwise.%
\end{tabular}%
\right. 
\end{equation*}%
We say that $\phi (m)$ is an \emph{effective allocation} under $m$. In
words, the effective allocation is $x$, if there are $I-1$ players who agree
on allocation $x$; otherwise, the effective allocation is the allocation
announced by agent 1.

The allocation rule $g$ is defined as follows: for each $m\in M$,

\begin{equation*}
g\left( m\right) =\frac{1}{I^{2}}\sum_{i\in \mathcal{I}}\sum_{j\in \mathcal{I%
}}g_{i,j}(m)\text{,} 
\end{equation*}%
and, 
\begin{equation*}
g_{i,j}(m)=\left[ e_{i,j}\left( m_{i},m_{j}\right) \left( \frac{1}{I}%
\sum_{k\in \mathcal{I}}y_{k}(m_{k,k}^{1})\right) \oplus \left(
1-e_{i,j}\left( m_{i},m_{j}\right) \right) x(\tilde{\theta},\phi
(m),m_{j,j}^{1})\right] . 
\end{equation*}
where $\{y_{k}(\theta _{k})\}_{\theta _{k}\in \Theta _{k}}$ are the dictator
lotteries for agent $k$ as defined in Lemma \ref{AM}. Given a message
profile $m,$ and a pair of agents $i$ and $j,$ we say that agent $j$ \emph{%
challenges agent} $i$ if and only if $m_{i}^{3}=\phi (m)$ and $%
x(m_{i}^{2},\phi (m),m_{j,j}^{1})\neq \phi (m)$, i.e., agent $i$'s reported
allocation is an effective one and agent $j$ challenges this effective
allocation. We define the $e_{i,j}$-function as follows: for each $m\in M$,

\begin{equation*}
e_{i,j}(m_{i},m_{j})=\left\{ 
\begin{array}{ll}
\varepsilon \text{,} & \text{if agent }j\text{ challenges agent }i\text{;}
\\ 
0\text{,} & \text{otherwise }.%
\end{array}%
\right. 
\end{equation*}%
Recall that the $e_{i,j}$-function in Section \ref{ar} for the case of SCFs
only depends on $m_{i}$ and $m_{j}$. In contrast, the $e_{i,j}$-function
here depends on the entire message profile, as the nature of the challenge
also depends on whether the allocation reported by agent $i$ is an effective
allocation or not.

Fix $i,j\in \mathcal{I},\varepsilon \in (0,1)$, and $m\in M$. Then, we
define 
\begin{equation*}
C_{i,j}^{\varepsilon }(m)\equiv \varepsilon \times \frac{1}{I}\sum_{k\in 
\mathcal{I}}y_{k}\left( m_{k,k}^{1}\right) \oplus \left( 1-\varepsilon
\right) \times x(m_{i}^{2},\phi (m),m_{j,j}^{1})\text{.} 
\end{equation*}%
For every message profile $m$ and agent $j$, we can choose $\varepsilon >0$
sufficiently small such that (i) $C_{i,j}^{\varepsilon }(m)$ does not
disturb the \textquotedblleft effectiveness\textquotedblright\ of agent $j$%
's challenge, i.e.,

\begin{eqnarray}
x(m_{i}^{2},\phi (m),m_{j}^{1}) &\neq &\phi (m)\Rightarrow  \notag \\
u_{j}(C_{i,j}^{\varepsilon }(m),m_{i,j}^{2}) &<&u_{j}(\phi (m),m_{i,j}^{2})%
\text{ and }u_{j}(C_{i,j}^{\varepsilon }(m),m_{j,j}^{1})>u_{j}(\phi
(m),m_{j,j}^{1})\text{.}  \label{Cbw}
\end{eqnarray}

\subsubsection{\label{CFmoney}Transfer Rule}

We now define the transfer rule. For every message profile $m\in M$ and
every agent $i\in \mathcal{I}$, we specify the transfer received by agent $i$
as follows: 
\begin{equation*}
\tau _{i}(m)=\sum_{j\neq i}\left[ sc_{i,j}(m)+\tau
_{i,j}^{1}(m_{i},m_{j})+\tau _{i,j}^{2}(m_{i},m_{j})+\tau
_{i,j}^{3}(m_{i},m_{j})\right]. 
\end{equation*}

For the transfer $sc_{i,j}$, we define it as a strict scoring rule, a
\textquotedblleft bet\textquotedblright\ on his opponent agent $j$
challenging agent $j$ himself. Thus, the transfers depend on the probability
reported by agent $i$ on the event agent $j$ being challenged, $c_{i,j}$ and
the one realized. We define the proper scoring rule explicitly as follows, 
\begin{equation*}
sc_{i,j}\left( m\right) =\left\{ 
\begin{array}{ll}
-(c_{i,j})^{2}-(1-c_{i,j})^{2}+2c_{i,j}-1\text{, } & \text{if agent }j\text{
challenges himself}; \\ 
-(c_{i,j})^{2}-(1-c_{i,j})^{2}+2(1-c_{i,j})-1\text{,} & \text{otherwise.}%
\end{array}%
\right. 
\end{equation*}%
%
%
%
%
%
%
%
%
%
%
%
%
%
%
%
%
%
%
%
%
%
%
%
%
%
% By design, it is easy to see that with belief $\sigma_j(B_{m^2_i,j}|m_i)$ for agent $i$, it is strictly better for agent $i$ to report $c^1_{i,j}=\sigma_j(B_{m^2_i,j}|m_i)$ rather than any other belief. 
% \begin{comment}
%     When $\sigma_j(B_{m^2_i,j}|m_i)=0$, we have $sc^1_{i,j}=0.$
% \end{comment}
\begin{eqnarray}
\tau _{i,j}^{1}\left( m_{i},m_{j}\right) &=&\left\{ 
\begin{tabular}{ll}
$-2\eta $ , & \text{if} $\text{agent }j\text{ challenges agent }i$; \\ 
$0$, & \text{otherwise.}%
\end{tabular}%
\ \right.  \label{Cbcha} \\
\tau _{i,j}^{2}\left( m_{i},m_{j}\right) &=&\left\{ 
\begin{tabular}{ll}
$-\varepsilon $ , & \text{if} $\text{agent }j\text{ challenges himself}$ 
\text{ and } $m_{i,j}^{1}\neq m_{j,j}^{1}$; \\ 
$0$, & \text{otherwise.}%
\end{tabular}%
\ \right.  \label{Cspotother} \\
\tau _{i,j}^{3}\left( m_{i},m_{j}\right) &=&\left\{ 
\begin{array}{ll}
-\eta \text{,} & \text{if }c_{j,i}>0\text{ and }m_{i,i}^{2}\not=m_{j,i}^{1}{;%
} \\ 
0\text{,} & \text{otherwise.}%
\end{array}%
\right.  \label{Cbacktotruth}
\end{eqnarray}%
Recall that $\eta >0$ is chosen to be greater than the maximal utility
difference from the outcome function $g\left( \cdot \right) ,$ and $%
\varepsilon $ is a small positive nubmer; see (\ref{D}).

In the rest of the proof , we fix $\theta $ as the true state and $\sigma $
as a correlated equilibrium of the game $\Gamma (\mathcal{M},\theta )$
throughout.

\subsubsection{Proof of Implementation}

Now, for each state $\tilde{\theta}\in \Theta $, and each allocation $x\in
F\left( \tilde{\theta}\right) ,$ we define the following set of agents: 
\begin{equation*}
\mathcal{J}\left( \tilde{\theta},x\right) \equiv \left\{ j\in \mathcal{I}:%
\mathcal{L}_{j}(x,\tilde{\theta}_{j})\cap \mathcal{SU}_{j}(x,\theta
_{j})=\varnothing \right\} \text{.} 
\end{equation*}

\begin{claim}
\label{Cwhistle}Given an arbitrary agent $i$ and message profile $m\in $supp$%
\left( \sigma \right) $ such that $m_{i}^{3}=\phi (m)$, we have $%
x(m_{i}^{2},m_{i}^{3},m_{j,j}^{1})\not=m_{i}^{3}\ $if and only if $j\not\in 
\mathcal{J}\left( m_{i}^{2},m_{i}^{3}\right) $.
\end{claim}

\begin{proof}
Fix agent $i\in \mathcal{I}$ and a message profile $m\in $supp$(\sigma )$.
We first prove the if-part. Suppose, on the contrary, that there exists some
agent $j\not\in \mathcal{J}\left( m_{i}^{2},m_{i}^{3}\right) $ such that $%
x(m_{i}^{2},m_{i}^{3},m_{j,j}^{1})=m_{i}^{3}$. Then, consider the deviation
from $m_{j}$ to $\tilde{m}_{j}$ such that $\tilde{m}_{j,j}^{1}=\theta _{j}$
while $m_{j}$ and $\tilde{m}_{j}$ are the same in every other dimension. By (%
\ref{Cbw}), it delivers a strictly better payoff for agent $j$ against $%
m_{-j}$ where $j\not\in \mathcal{J}\left( m_{i}^{2},m_{i}^{3}\right) $. At
the same time, by Lemmas \ref{AM} and \ref{BCC}, the deviation from $m_{j}$
to $\tilde{m}_{j}$ generates no payoff loss for agent $j$ against any $%
m_{-j}^{\prime }\not=m_{-j}$. Thus, the deviation $\tilde{m}_{j}$ is
profitable, which contradicts the hypothesis that $\sigma $ is an
equilibrium of the game $\Gamma (\mathcal{M},\theta )$.

Next, we prove the only-if-part. Suppose, on the contrary, that there exists
some agent $j\in \mathcal{J}\left( m_{i}^{2},m_{i}^{3}\right) $ such that $%
x(m_{i}^{2},m_{i}^{3},m_{j,j}^{1})\not=m_{i}^{3}.$ Since $j\in \mathcal{J}%
\left( m_{i}^{2},m_{i}^{3}\right) $, we must have $m_{j,j}^{1}\neq \theta
_{j}$. Define $\tilde{m}_{j}$ as a deviation which is identical to $m_{j}$
except that $\tilde{m}_{j,j}^{1}=\theta _{j}\neq m_{j,j}^{1}$. Then we have $%
x(m_{i}^{2},m_{i}^{3},\theta _{j})=m_{i}^{3}$ since $j\in \mathcal{J}\left(
m_{i}^{2},m_{i}^{3}\right) .$ By (\ref{Cbw}), $\tilde{m}_{j}$ generates a
strictly better payoff for agent $j$ than $m_{j}$ against $m_{-j}$. By
Lemmas \ref{AM} and \ref{BCC}, we also know that agent $j$'s payoff
generated by $\tilde{m}_{j}$ is at least as good as that generated by $m_{j}$
against any $m_{-j}^{\prime }\neq m_{-j}$. Hence, $\tilde{m}_{j}$
constitutes a profitable deviation, which contradicts the hypothesis that $%
\sigma $ is an equilibrium of the game $\Gamma (\mathcal{M},\theta )$.
\end{proof}

\begin{claim}
\label{Cother} No one challenges an allocation announced by any other agent,
i.e., for any pair of agents $i,j\in \mathcal{I}$ with $i\neq j$ and any $%
m\in $supp$\left( \sigma \right) $ such that $m_{i}^{3}=\phi (m)$, $%
x(m_{i}^{2},m_{i}^{3},m_{j,j}^{1})=m_{i}^{3}$.
\end{claim}

\begin{proof}
Suppose to the contrary that there exist $i,j\in \mathcal{I}$ with $i\neq j$%
, $m\in \mbox{supp}(\sigma )$ such that such that $m_{i}^{3}=\phi (m)$, $%
x(m_{i}^{2},m_{i}^{3},m_{j,j}^{1})\not=m_{i}^{3}$. By Claim \ref{Cwhistle}, $%
j\not\in \mathcal{J}\left( m_{i}^{2},m_{i}^{3}\right) $. By Claim \ref%
{Cwhistle}, we know that for any agent $j\not\in \mathcal{J}\left(
m_{i}^{2},m_{i}^{3}\right) ,$ we have $x(m_{i}^{2},m_{i}^{3},\hat{m}%
_{j,j}^{1})\not=m$, for every $\hat{m}_{-i}$ such that $\sigma (m_{i},\hat{m}%
_{-i})>0$. In addition, $\hat{m}_{j,j}^{1}=\theta _{j}$ for every $j$ such
that $j\not\in \mathcal{J}\left( m_{i}^{2},m_{i}^{3}\right) .$ Hence, from
agent $i$'s perspective, conditional on playing $m_{i},$ he is challenged
with probability $1.$ Thus, agent $i$ is penalized $2\eta $ due to (\ref%
{Cbcha})). Consider a deviation to $\tilde{m}_{i}$ which is the same as $%
m_{i}$ except that $\tilde{m}_{i}^{2}=\theta $. Note that from agent $i$'s
perspective, dictator lotteries are triggered with probability one, hence,
all the agents report the truth in the first reports. Then, agent $i$ avoids
paying the penalty $2\eta $ for being challenged, while the potential loss
from allocation is bounded by $\eta ,$ the loss from (\ref{backtotruth}) is
bounded by $\eta .$ Therefore, it is a profitable deviation. This
contradicts to that $\sigma $ is an equilibrium.
\end{proof}

\begin{claim}
\label{Cbet}For any pair of agents $i$ and $j$ such that $i\neq j$, message
profile $m\in $supp$\left( \sigma \right) $ such that $m_{j}^{3}=\phi (m)$,
if $x(m_{j}^{2},m_{j}^{3},m_{j,j}^{1})\neq m_{j}^{3}$, then we have $%
c_{i,j}>0$, $m_{i,j}^{1}=\theta _{j}$ and $m_{i,i}^{1}=\theta _{i}$.
\end{claim}

\begin{proof}
First, we fix an arbitrary correlated equilibrium $\sigma $. Given $m_{i}$
we collect all the opponents' message profile together with $m_{i}$ agent $i$
knows that agent $j$ is challenged by himself: 
\begin{equation*}
E_{j}(m_{i})=\{(m_{i},\tilde{m}_{-i}):g_{j,k}(m_{i},\tilde{m}_{-i})\neq 
\tilde{m}_{j}^{3}=\phi (m_{i},\tilde{m}_{-i})\text{ for some }k\text{ and }%
\sigma (m_{i},\tilde{m}_{-i})>0.\} 
\end{equation*}%
From the condition in Claim \ref{Cbet}, we know that $E_{j}(m_{i})\neq
\varnothing $. By the proper scoring rule $sc_{i,j}$, $m_{i,j}^{4}=\sum_{%
\tilde{m}_{-i}\in E_{j}(m_{i})}\sigma _{-i}(\tilde{m}_{-i}|m_{i})$.
Conditional on $m_{i}$, dictator lotteries are triggered with positive
probability; hence by Lemmas \ref{AM} and \ref{BCC}, $m_{i,i}^{1}=\theta
_{i} $. In addition, when agent $j$ is challenged by himself, we have $%
m_{j,j}^{1}=\theta _{j}$. Hence, according to $\tau _{i,j}^{2}$ (see (\ref%
{Cspotother})), we have $m_{i,j}^{1}=\theta _{j}.$ 
% ??DEL?It is easy to see that compared to any $\tilde{m_{i}}$ which is
% identical to $m_{i}$ but $\tilde{m}_{i,j}^{1}\neq \theta _{j}$, $m_{i}$
% provides a better payoff against $m_{-i}$, while is at least as good as $%
% \tilde{m_{i}}$ against any other $\tilde{m}_{-i}$.??
\end{proof}

\begin{claim}
\label{Cnochallenge}For every agent $j$, message profile $m\in $supp$\left(
\sigma \right) $ such that $m_{j}^{3}=\phi (m)$, we have $%
x(m_{j}^{2},m_{j}^{3},m_{k,k}^{1})=m_{j}^{3}$ for every agent $k$.
\end{claim}

\begin{proof}
By Claim \ref{Cother}, we know that for every $k\neq j$, Claim \ref%
{Cnochallenge} holds. It remains to show that Claim \ref{Cnochallenge} holds
for $k=j$. Suppose there exists $m\in $supp$\left( \sigma \right) $, such
that $m_{j}^{3}=\phi (m)$, if $x(m_{j}^{2},m_{j}^{3},m_{j,j}^{1})\neq
m_{j}^{3}$. By Lemmas \ref{AM} and \ref{BCC}, $m_{j,j}^{1}=\theta _{j}$. Now
we show that for agent $j$ who reports $m_{j}$ with $m_{j,j}^{2}\not=\theta
_{j}$, it is strictly better for agent $j$ to deviate to $\tilde{m}_{j}$
which is identical to $m_{j}$ but $\tilde{m}_{j,j}^{2}=\theta _{j}$.
Specifically, due to Claim \ref{Cbet}, for every $\tilde{m}_{-j}$ such that $%
\sigma (m_{j},\tilde{m}_{-j})>0$, for every agent $i\neq j$, we have $\tilde{%
m}_{i,j}^{4}>0$, $\tilde{m}_{i,j}^{1}=\theta _{j}$. By $\tau ^{3}$ in (\ref%
{Cbacktotruth}), we know that the deviation is profitable. Thus, it
contradicts the hypothesis that $x(m_{j}^{2},m_{j}^{3},m_{j,j}^{1})\neq
m_{j}^{3}$.
\end{proof}

\begin{claim}
For every $m\in\text{supp}(\sigma)$, $g(m)=f(\theta)$, and $\tau_{i,j}(m)=0$
for every agent $i$ and $j$.
\end{claim}

\begin{proof}
By Claim \ref{Cwhistle} and Claim \ref{Cnochallenge}, we know that for every 
$m\in \text{supp}(\sigma )$, $g(m)\in F(\theta )$, and $\tau _{i,j}^{1}=0$
for every agent $i$ and $j$. Due to the construction of the proper scoring
rule $sc$, we know that $m_{i,j}^{4}=0$ for every agent $i$ and $j$. Hence $%
\tau _{i,j}^{2}=\tau _{i,j}^{3}=0$ for every agent $i$ and $j$. Hence, we
achieve implementation with no transfers incurred.
\end{proof}

\subsection{Related Literature}

\cite{aumann1974subjectivity} first proposed the concept of correlated
equilibrium (CE) as a generalization of the concept of independent
randomization among agents and \cite{hart1989existence} were the first to
offer a direct proof of existence for CE. An important feature of correlated
equilibria is that they can be derived as the result of several adaptive
procedures. The first of these were proposed by \cite{foster1997calibrated}
in which agents forecast other's play and play best responses to this
forecast - this eventually yields a correlated equilibrium. \cite%
{fudenberg1999conditional} show that a variant of fictitious play which
satisfies a condition called conditional universal consistency also yields a
CE. We build on the work of \cite{HartMasColell2000} who show (for finite N
player games) that the regret matching heuristic yields CE of the underlying
game. \cite{stoltz2007learning} generalize this result to infinite but
convex and compact strategy sets. In principle our implementation result
applies to each of these adaptive procedures and we select regret-matching
for concreteness.

\cite{maskin99} pioneered the concept of implementation of a social choice
function (SCF) in Nash equilibria and derived the appropriate
characterization, showing that a monotonicity condition was necessary and
almost sufficient for implementability in Nash equilibria. Later results
studied implementation in various solution concepts such as subgame perfect
Nash equilibria \cite{MR88}, iteratively undominated strategies \cite{AM92},
and rationalizability \cite{bergemann2011rationalizable}. Classical
mechanisms in this literature often use a device called an ``integer game''
which allows the designer to eliminate many undesirable message profiles by
augmenting them with an integer so that the agent with the highest integer
can dictatorially pick the outcome. This creates a ``race to the top'' and
thus such games possess no equilibria. These ``tail chasing devices'' are
critiqued in both \cite{Jackson92} and \cite{Moore92}. Further, mechanisms
which use integer games cannot be simulated (since they have an infinite
message space) although a variant called a modulo game (with a finite
message space) can be implemented in simulations. However, a modulo game
yields undesirable mixed strategy equilibria, so that finite mechanisms are
best suited for simulations. Recent papers have attempted to place the
theory of implementation on a more theoretically sound footing by eschewing
the use of integer games \cite{MAM,chen2021rationalizable,Fehr2021}, but
general results for implementation in correlated equilibria have not yet
been derived. This paper presents, to our knowledge, the first
characterization of SCFs which can be implemented in CE and also presents a
well behaved mechanism which does not use integer games.

Most recently, \cite{pei2025robust} show that the mechanism in \cite%
{chen2021rationalizable} robustly implements SCFs satisfying
maskin-monotonicity* by invoking Proposition 3.2 in \cite%
{kajii1997robustness} which shows that a unique correlated equilibrium must
be a robust equilibrium. They also present mechanisms which robustly
implement SCFs which do not satisfy maskin monotonicity by relying on costly
information. Relatedly, the mechanism proposed here implements SCFs
satisfying maskin monotonicty (which is weaker than maskin monotonicity*) as
the outcome of a unique correlated equilibrium of the game induced by the
mechanism, and thus also robustly implements. % Head 1

\section{\label{SIM}Simulation: Bilateral Trade}

\subsection{\label{SIMENV}Environment}

A seller $S$ has an object for sale to a buyer $B$. The seller can have a low or a high cost of production while the buyer may have a low or high valuation for the product. The costs for the seller are given by $c^H = 8$ and $c^L = 2$, while the valuations for the buyer are given by $v^H = 20$ and $v^L = 12$. Since each of the players in the game can have one of two types, there are four possible states of the world. The designer can impose
transfers and hence the set of outcomes $A$ is the set of triplets $\left(
q,t_{B},t_{S}\right) $ with $q\in \left[ 0,1\right] $ representing the
amount of the good being traded, $t_{B}$ is the price paid by $B$ and $t_{S}$
is the payment received by $S$. For any outcome $\left( q,t_{B},t_{S}\right)
,$ $B$'s utility is $u_{B}=qv+t_{B}$ when the good quality is $v,$ and the
seller's utility is $u_{S}=t_{S}-qc$. The desired allocation (SCF) is shown in Table \ref{tab:original_scf}.
\begin{table}[h!]
\centering
\caption{Social Choice Function (SCF) Outcomes by State}
\label{tab:original_scf}
\begin{tabular}{cc ccc}
\toprule
\multicolumn{2}{c}{\textbf{State of the World}} & \multicolumn{3}{c}{\textbf{SCF Outcome}} \\
\cmidrule(r){1-2} \cmidrule(l){3-5}
\textbf{Buyer Type} & \textbf{Seller Type} & \textbf{Quantity ($q$)} & \textbf{Buyer Payment ($t_B$)} & \textbf{Seller Payment ($t_S$)} \\
\midrule
L & L & 1 & -6 & 6 \\
H & H & 1 & -10 & 10 \\
H & L & 1 & -10 & 10 \\
L & H & 1 & -10 & 10 \\
\bottomrule
\end{tabular}
\end{table}
\subsection{\label{SIMALLOC}Test Allocations}
For the SCF to be implementable, it must satisfy Maskin Monotonicity. This condition ensures that if agents misreport the state of the world, at least one agent (a ``whistle-blower'') will find it in their interest to challenge the lie. 

Specifically, if the true state is $\theta$ but agents report a lie $\theta'$, there must be a test allocation from a challenge scheme, denoted $x(\theta', \theta_i)$ as in Lemma \ref{BC}, for a whistle-blowing agent $i$ (whose type differs between $\theta$ and $\theta'$) such that two conditions are met:
\begin{enumerate}
    \item \textbf{No False Alarms:} In the lie state $\theta'$, the agent weakly prefers the original outcome to the test allocation. That is, $u_i(f(\theta'), \theta'_i) \geq u_i(x(\theta', \theta_i), \theta'_i)$.
    \item \textbf{Incentive to Expose Lie:} In the true state $\theta$, the agent strictly prefers the test allocation to the outcome from the lie. That is, $u_i(x(\theta', \theta_i), \theta_i) > u_i(f(\theta'), \theta_i)$.
\end{enumerate}

For the bilateral trade simulation, the challenge scheme $\{x(\theta', \theta_i)\}$ is defined with the following key test allocations:
\begin{itemize}
    \item When the Buyer is the whistle-blower (true type $v^L$, lie type $v^H$), the test allocation is $x(\theta', v^L) = (0.5, -2, 2)$.
    \item When the Seller is the whistle-blower (true type $c^H$, lie type $c^L$), the test allocation is $x(\theta', c^L) = (0.5, -3, 3)$.
\end{itemize}

To demonstrate their validity, we will show that for the true state of the world being \textbf{(L, H)}, a whistle-blower exists for every possible lie.

\subsubsection{Lie State: (L, L)}
If agents report $\theta'=(L,L)$, the outcome is $f(L,L)=(1, -6, 6)$. The Seller's true type ($c^H$) differs from the lie ($c^L$), making them the whistle-blower. We check the two conditions for the Seller using the test allocation $x((L,L), c^H) = (0.5, -3, 3)$.
\begin{itemize}
    \item \textbf{Condition 1 (No False Alarms):} We check if $u_S(f(L,L), c^L) \geq u_S(x((L,L), c^H), c^L)$.
    \begin{itemize}
        \item $u_S(f(L,L), c^L) = 6 - (2 \times 1) = 4$.
        \item $u_S(x((L,L), c^H), c^L) = 3 - (2 \times 0.5) = 2$.
        \item The condition $4 \geq 2$ holds.
    \end{itemize}
    \item \textbf{Condition 2 (Incentive to Expose):} We check if $u_S(x((L,L), c^H), c^H) > u_S(f(L,L), c^H)$.
    \begin{itemize}
        \item $u_S(x((L,L), c^H), c^H) = 3 - (8 \times 0.5) = -1$.
        \item $u_S(f(L,L), c^H) = 6 - (8 \times 1) = -2$.
        \item The condition $-1 > -2$ holds.
    \end{itemize}
\end{itemize}
Both conditions are met, so the Seller can act as a whistle-blower.

\subsubsection{Lie State: (H, H)}
If agents report $\theta'=(H,H)$, the outcome is $f(H,H)=(1, -10, 10)$. The Buyer's true type ($v^L$) differs from the lie ($v^H$), making them the whistle-blower. We check the two conditions for the Buyer using the test allocation $x((H,H), v^L) = (0.5, -2, 2)$.
\begin{itemize}
    \item \textbf{Condition 1 (No False Alarms):} We check if $u_B(f(H,H), v^H) \geq u_B(x((H,H), v^L), v^H)$.
    \begin{itemize}
        \item $u_B(f(H,H), v^H) = (20 \times 1) - 10 = 10$.
        \item $u_B(x((H,H), v^L), v^H) = (20 \times 0.5) - 2 = 8$.
        \item The condition $10 \geq 8$ holds.
    \end{itemize}
    \item \textbf{Condition 2 (Incentive to Expose):} We check if $u_B(x((H,H), v^L), v^L) > u_B(f(H,H), v^L)$.
    \begin{itemize}
        \item $u_B(x((H,H), v^L), v^L) = (12 \times 0.5) - 2 = 4$.
        \item $u_B(f(H,H), v^L) = (12 \times 1) - 10 = 2$.
        \item The condition $4 > 2$ holds.
    \end{itemize}
\end{itemize}
Both conditions are met, so the Buyer can act as a whistle-blower.

\subsubsection{Lie State: (H, L)}
If agents report $\theta'=(H,L)$, the outcome is $f(H,L)=(1, -10, 10)$. Here, both players' true types differ from their reported types. We only need one to be a whistle-blower. Let's check the Buyer (type change H $\rightarrow$ L) using the test allocation $x((H,L), v^L) = (0.5, -2, 2)$. The calculations are identical to the previous case, as the lie outcome is the same.
\begin{itemize}
    \item \textbf{Condition 1} ($10 \geq 8$) holds.
    \item \textbf{Condition 2} ($4 > 2$) holds.
\end{itemize}
The Buyer can act as a whistle-blower. Since a valid whistle-blower exists for every possible lie, the SCF is implementable for the true state of (L, H).

\subsection{\label{SIMDICTATOR}Dictator Lotteries}
The mechanism also employs ``dictator lotteries'' to ensure that agents have a strict incentive to report their own type truthfully in certain parts of their message[cite: 50]. A set of dictator lotteries $\{y_i(\theta_i)\}$ is valid if an agent $i$ with true type $\theta_i$ strictly prefers the lottery associated with their true type over the lottery for any other type $\theta'_i$[cite: 51]. That is, for any $\theta_i \neq \theta'_i$, it must be that $u_i(y_i(\theta_i), \theta_i) > u_i(y_i(\theta'_i), \theta_i)$.

The lotteries used in the simulation are as follows:

\subsubsection{Buyer's Dictator Lotteries}
The two dictator lotteries available for the buyer are:
\begin{itemize}
    \item If buyer's reported type is H: $y_B(H) = (1, -15, 15)$.
    \item If buyer's reported type is L: $y_B(L) = (0, 0, 0)$.
\end{itemize}
\textbf{Validity Check:}
\begin{itemize}
    \item When the Buyer's true type is \textbf{H} ($v^H=20$):
    \begin{itemize}
        \item Utility from true report $y_B(H)$: $u_B = (20 \times 1) - 15 = 5$.
        \item Utility from false report $y_B(L)$: $u_B = (20 \times 0) + 0 = 0$.
        \item The condition $5 > 0$ holds.
    \end{itemize}
    \item When the Buyer's true type is \textbf{L} ($v^L=12$):
    \begin{itemize}
        \item Utility from true report $y_B(L)$: $u_B = (12 \times 0) + 0 = 0$.
        \item Utility from false report $y_B(H)$: $u_B = (12 \times 1) - 15 = -3$.
        \item The condition $0 > -3$ holds.
    \end{itemize}
\end{itemize}
The lotteries are valid for the Buyer.

\subsubsection{Seller's Dictator Lotteries}
The two dictator lotteries available for the seller are:
\begin{itemize}
    \item If seller's reported type is H: $y_S(H) = (0, 0, 0)$.
    \item If seller's reported type is L: $y_S(L) = (1, -4, 4)$.
\end{itemize}
\textbf{Validity Check:}
\begin{itemize}
    \item When the Seller's true type is \textbf{H} ($c^H=8$):
    \begin{itemize}
        \item Utility from true report $y_S(H)$: $u_S = 0 - (8 \times 0) = 0$.
        \item Utility from false report $y_S(L)$: $u_S = 4 - (8 \times 1) = -4$.
        \item The condition $0 > -4$ holds.
    \end{itemize}
    \item When the Seller's true type is \textbf{L} ($c^L=2$):
    \begin{itemize}
        \item Utility from true report $y_S(L)$: $u_S = 4 - (2 \times 1) = 2$.
        \item Utility from false report $y_S(H)$: $u_S = 0 - (2 \times 0) = 0$.
        \item The condition $2 > 0$ holds.
    \end{itemize}
\end{itemize}
The lotteries are also valid for the Seller.

\subsection{Regret Minimization}

The reader is referred to \cite{HartMasColell2000}, \cite{HartMasColell2003}%
, and \cite{Hart2005} for a comprehensive description of the regret
minimization heuristic and a theoretical proof of the claim that when
players use the regret minimization heuristic to play a game, the long run
distribution of the play eventually converges to a correlated equilibrium of
the game. For completeness, we provide a very brief overview of the regret
minimization heuristic here.

Consider a player $i$ with a set of messages $M_{i}$ in a repeated game
setting over $T$ periods. Let $m_{i}^{t}\in M_{i}$ denote the message chosen
by player $i$ in period $t$, and $m_{-i}^{t}$ denote the messages chosen by
all other players in period $t$. Let $u_{i}(m_{i},m_{-i})$ be the payoff
function of player $i$.\footnote{%
The reader is reminded that the utility function $u(m_{i},m_{-i})$ comprises
the utility from the outcome, denoted by $v(g(m_{i},m_{-i}))$, and any
transfers the mechanism might prescribe.}

\subsubsection{Regret Calculation}

For each message $m_{i}^{\prime }\in M_{i}$, the regret of not having played 
$m_{i}^{\prime }$ up to period $T$ is defined as: 
\begin{equation*}
R_{i}^{T}(m_{i}^{\prime })=\frac{1}{T}\sum_{t=1}^{T}\left[
u_{i}(m_{i}^{\prime },m_{-i}^{t})-u_{i}(m_{i}^{t},m_{-i}^{t})\right] , 
\end{equation*}%
where $u_{i}(m_{i}^{\prime },m_{-i}^{t})$ is the hypothetical payoff player $%
i$ would have received had they played $m_{i}^{\prime }$ in period $t$,
while holding $m_{-i}^{t}$ fixed.

\subsubsection{Strategy Update Rule}

Based on the calculated regrets $\{R_{i}^{T}(m_{i}^{\prime }):m_{i}^{\prime
}\in M_{i}\}$, the player updates their probability distribution over
messages. A common rule is to increase the probability of messages with
positive regret. For example: 
\begin{equation*}
p_{i}^{T+1}(m_{i}^{\prime })=\frac{\max \{R_{i}^{T}(m_{i}^{\prime }),0\}}{%
\sum_{m_{i}^{\prime \prime }\in M_{i}}\max \{R_{i}^{T}(m_{i}^{\prime \prime
}),0\}}, 
\end{equation*}%
where $p_{i}^{T+1}(m_{i}^{\prime })$ denotes the probability of choosing
message $m_{i}^{\prime }$ in the next period.

If all regrets are non-positive, the player may revert to a uniform or prior
distribution over messages.

\subsubsection{Regret Minimization Objective}

The player's objective is to minimize their average regret over time,
defined as: 
\begin{equation*}
\overline{R}_{i}^{T}=\max_{m_{i}^{\prime }\in M_{i}}R_{i}^{T}(m_{i}^{\prime
}). 
\end{equation*}%
A regret-minimizing algorithm ensures that $\overline{R}_{i}^{T}\rightarrow
0 $ as $T\rightarrow \infty $, implying no significant regret for not having
played any single message.

\subsection{Simulation Details and Results}

For ease of comparison across mechanisms, in the following simulations, the
state of the world is set to $(L,H)$ and then two agents (the buyer
and the seller) play the game induced by the mechanism detailed above using
regret matching. The test allocations are chosen as detailed in Section \ref%
{SIMALLOC}. The initial strategies involve randomizing uniformly over the
available messages. In about $200$ iterations, the dynamics are seen to
have converged to the correlated equilibrium of the game, i.e. truthtelling, and
no self-challenges. A sample simulation result is graphically represented in Figure \ref{fig:simresce}. 

As seen in Figure \ref{fig:simresce}, the players gradually converge to
playing the correlated equilibrium strategy with large probabilities despite
only following the regret minimizing heuristic. Since the mechanism
implements in correlated equilibria, this also yields the socially desirable
outcome with a large probability. The panels in the bottom half of the
figure show the evolution of the gains from trade with the growth of
transfers induced by the mechanism, i.e. any transfers other than those
required by the SCF itself.

\subsection{Comparisons with other mechanisms}

In \cite{HartMasColell2003}, the authors note that ``It is notoriously
difficult to formulate sensible adaptive dynamics that guarantee convergence
to Nash equilibrium''. \cite{MAM} provides a mechanism which can be used to
implement the SCF studied in this section in Nash equilibrium. We implement this mechanism (which we call the MAM mechanism) alongside regret dynamics using the same transfer scaling as the CE mechanism. The results are shown in Figure \ref{fig:simresmam}.

The mechanism takes longer (typically 2500 iterations or more) to converge to the correct actions, and incurs significantly more transfers in the process. Further, since it takes longer to find the optimal strategies, the social surplus is also significantly lower.

Allowing the dictator lotteries within the mechanism in \cite{MAM} to always be on with
probability $\epsilon$ provides an alternate mechanism in the spirit of \cite%
{AM92} which allows their mechanism to implement in correlated equilibria as
well, although the implementation is virtual. We simulate this mechanism (which we call the AM92 mechanism)
as well and present the results in Figure \ref{fig:simresam}.
This mechanism takes even longer to implement (of the order of 5000 iterations), and owing to the fact that the implementation is virtual, the dictator lotteries yield transfers with a small probability throughout the simulation, so that the transfers can be seen to dominate the social surplus. We conjecture that convergence
is likely to be even slower with the actual mechanism in \cite{AM92} for a
large number of rounds (K) since it would involve a larger message space.
Further, the message space of the mechanism in \cite{AM92} depends on the
number of rounds ($K$) which depends upon the degree of virtualness, $%
\epsilon$, whereas our mechanism remains the same regardless of the choice
of $\epsilon$.
\begin{figure}[!ht]
\centering
\includegraphics[scale = 0.6]{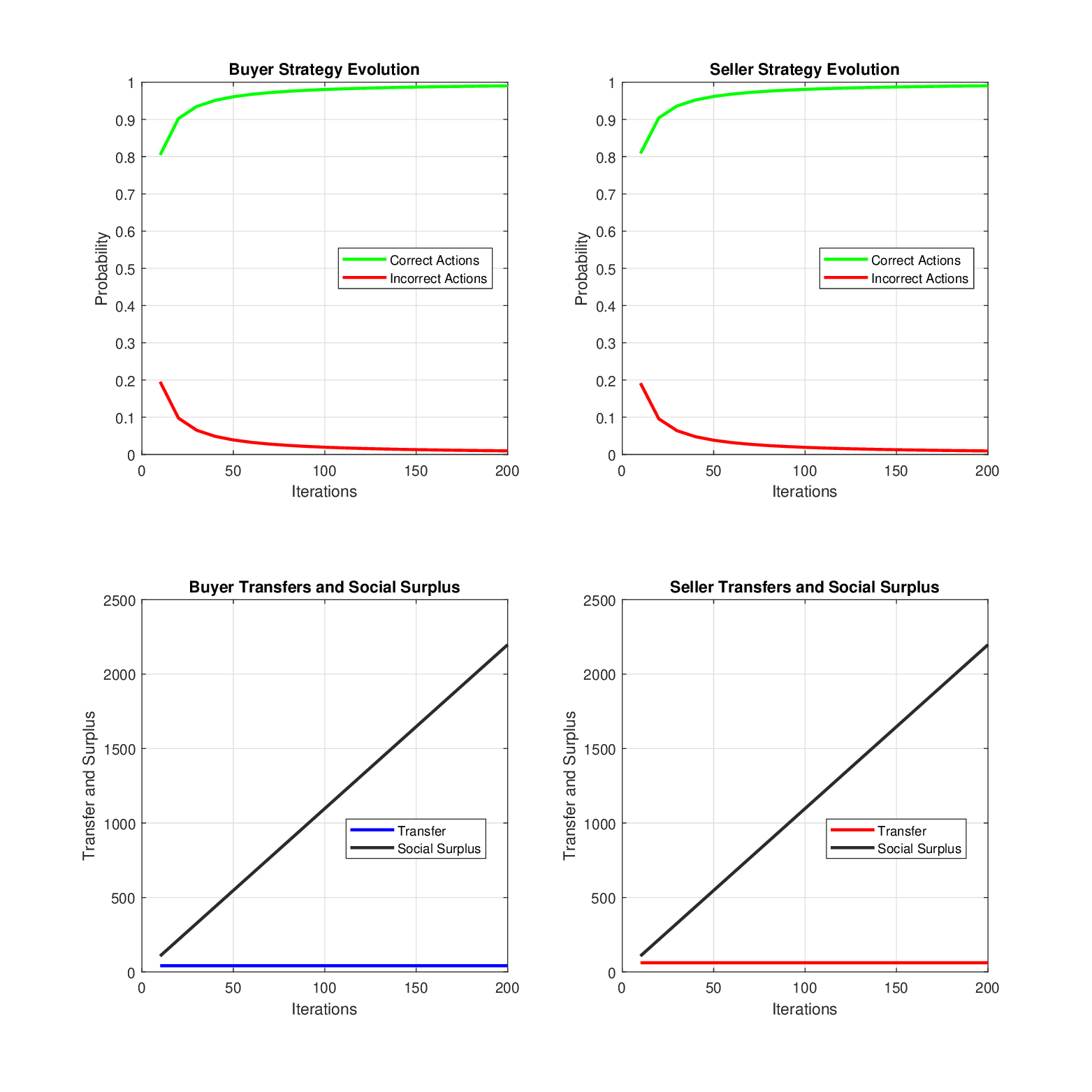}
\caption{Simulation Results (CE Mechanism)}
\label{fig:simresce}
\end{figure}
\begin{figure}[!ht]
\centering
\includegraphics[scale = 0.6]{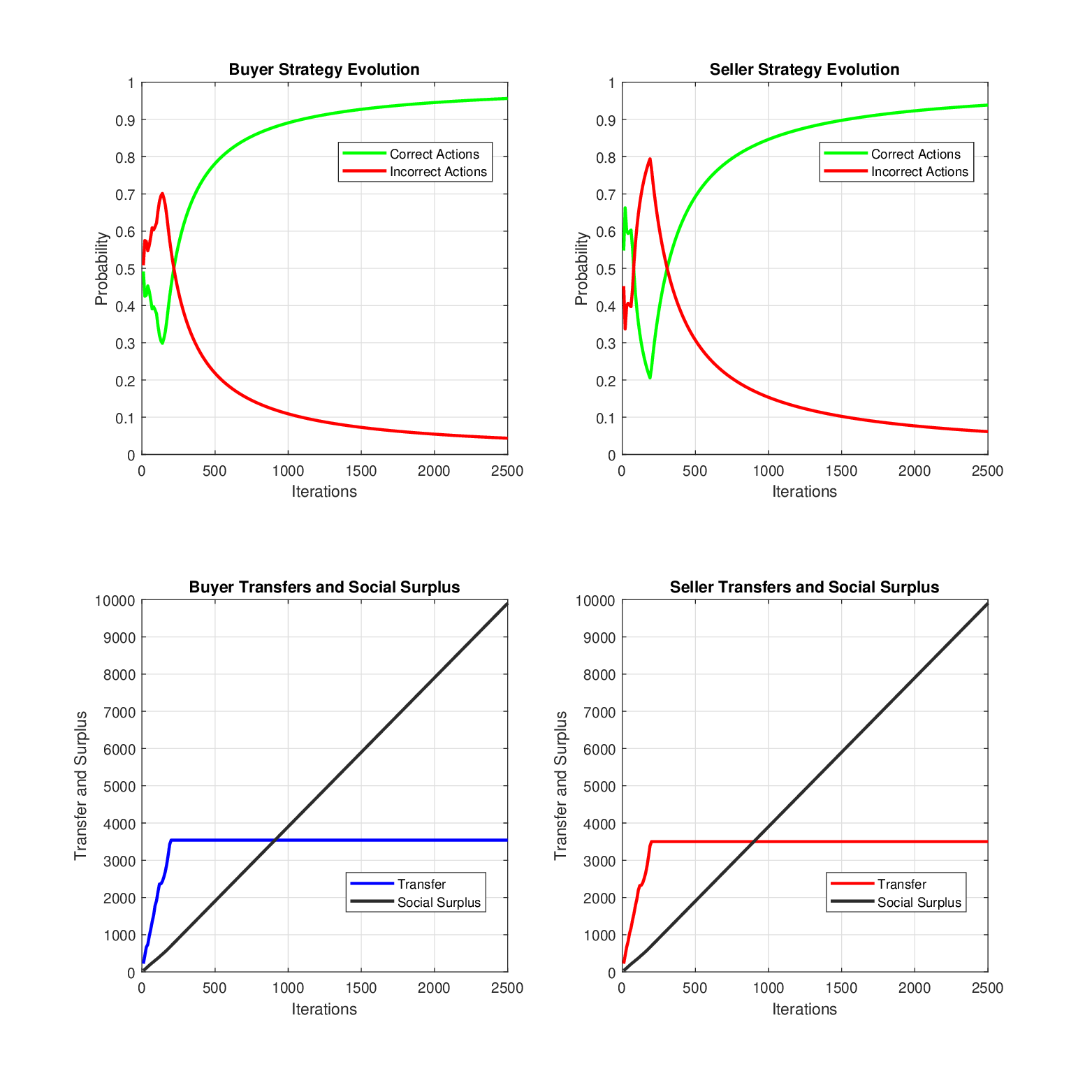}
\caption{Simulation Results (MAM Mechanism)}
\label{fig:simresmam}
\end{figure}
\begin{figure}[!ht]
\centering
\includegraphics[scale = 0.6]{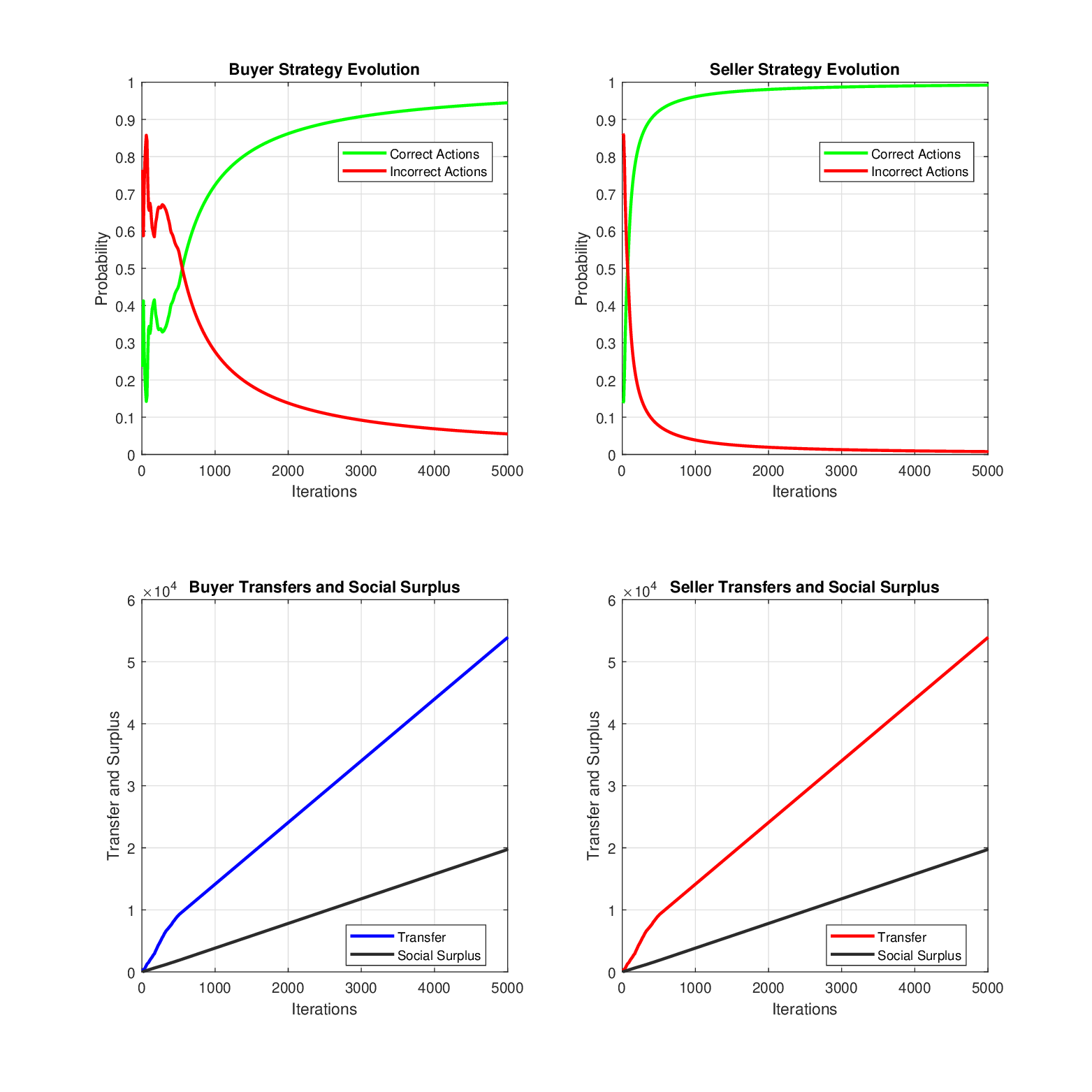}
\caption{Simulation Results (AM92 Mechanism)}
\label{fig:simresam}
\end{figure}
\bibliographystyle{ACM-Reference-Format}
\bibliography{sample-bibliography}

\section{Appendix}

The proof of Lemma\ref{BC}.

\begin{proof}
\label{proofofBC} Consider a challenge scheme $\bar{x}(\cdot ,\cdot )$.
First, we show that we can modify $\bar{x}(\cdot ,\cdot )$ into a new
challenge scheme $x(\cdot ,\cdot )$ such that 
\begin{equation}
x(\tilde{\theta},\theta _{i})\neq f(\tilde{\theta})\text{ and }x(\tilde{%
\theta},\theta _{i}^{\prime })\neq f(\tilde{\theta})\Rightarrow u_{i}(x(%
\tilde{\theta},\theta _{i}),\theta _{i})\geq u_{i}(x(\tilde{\theta},\theta
_{i}^{\prime }),\theta _{i})\text{.}  \label{b1}
\end{equation}%
To construct $x(\cdot ,\cdot )$, for each player $i$, we distinguish two
cases: (a) if $\bar{x}(\tilde{\theta},\theta _{i})=f(\tilde{\theta})$ for
all $\theta _{i}\in \Theta _{i}$, then set $x(\tilde{\theta},\theta _{i})=%
\bar{x}(\tilde{\theta},\theta _{i})=f(\tilde{\theta})$; (b) if $\bar{x}(%
\tilde{\theta},\theta _{i})\neq f(\tilde{\theta})$ for some $\theta _{i}\in
\Theta _{i}$, then define $x(\tilde{\theta},\theta _{i})\ $as the most
preferred allocation of type $\theta _{i}$ in the finite set 
\begin{equation*}
X(\tilde{\theta})=\left\{ \bar{x}(\tilde{\theta},\theta _{i}^{\prime
}):\theta _{i}^{\prime }\in \Theta _{i}\text{ and }\bar{x}(\tilde{\theta}%
,\theta _{i}^{\prime })\neq f(\tilde{\theta})\right\} . 
\end{equation*}%
Since $\bar{x}(\tilde{\theta},\theta _{i}^{\prime })\in \mathcal{L}_{i}(f((%
\tilde{\theta}),\tilde{\theta}_{i})$, we have $u_{i}(x(\tilde{\theta},\theta
_{i}),\tilde{\theta}_{i})\leq u_{i}(f(\tilde{\theta}),\tilde{\theta}_{i})$;
moreover, since $x(\tilde{\theta},\theta _{i})\ $as the most preferred
allocation of type $\theta _{i}$ in $X(\tilde{\theta})$ and $\bar{x}(\tilde{%
\theta},\theta _{i})\in \mathcal{SU}_{i}(f(\tilde{\theta}),\theta _{i})$, it
follows that $u_{i}(x(\tilde{\theta},\theta _{i}),\theta _{i})>u_{i}(f( 
\tilde{\theta}) ,\theta _{i})$. In other words, $x(\cdot ,\cdot )$ remains a
challenge scheme. Moreover, $x(\cdot ,\cdot )$ satisfies (\ref{b1}) by
construction.

Next, for each state $\tilde{\theta}$ and type $\theta _{i}$, we show that $%
x(\cdot ,\cdot )$ satisfies (\ref{b}). We proceed by considering the
following two cases. First, suppose that $x(\tilde{\theta},\theta _{i})\neq
f(\tilde{\theta})$. Then, by (\ref{b1}), it suffices to consider type $%
\theta _{i}^{\prime }$ with $x(\tilde{\theta},\theta _{i}^{\prime })=f(%
\tilde{\theta})$. Since $x(\tilde{\theta},\theta _{i}^{\prime })=f(\tilde{%
\theta})$ and $x(\tilde{\theta},\theta _{i})\neq f(\tilde{\theta})$, then it
follows from $x(\tilde{\theta},\theta _{i})\in \mathcal{SU}_{i}(f(\tilde{%
\theta}),\theta _{i})$ that $u_{i}(x(\tilde{\theta},\theta _{i}),\theta
_{i})>u_{i}(x(\tilde{\theta},\theta _{i}^{\prime }),\theta _{i})$. Hence, (%
\ref{best-C}) holds. Second, suppose that $x(\tilde{\theta},\theta _{i})=f(%
\tilde{\theta})$. Then, it suffices to consider type $\theta _{i}^{\prime }$
with $x(\tilde{\theta},\theta _{i}^{\prime })\neq f(\tilde{\theta})$. Since $%
x(\tilde{\theta},\theta _{i})=f(\tilde{\theta})$, we have $\mathcal{L}_{i}(f(%
\tilde{\theta}),\tilde{\theta}_{i})\cap \mathcal{SU}_{i}(f(\tilde{\theta}%
),\theta _{i})=\varnothing $. Moreover, $x(\tilde{\theta},\theta
_{i}^{\prime })\neq f(\tilde{\theta})$ implies that $x(\tilde{\theta},\theta
_{i}^{\prime })\in \mathcal{L}_{i}(f(\tilde{\theta}),\tilde{\theta}_{i})$.
Hence, we must have $x(\tilde{\theta},\theta _{i}^{\prime })\notin $ $%
\mathcal{SU}_{i}(f(\tilde{\theta}),\theta _{i})$. That is, $u_{i}(x(\tilde{%
\theta},\theta _{i}),\theta _{i})\geq u_{i}(x(\tilde{\theta},\theta
_{i}^{\prime }),\theta _{i}),$ i.e., (\ref{best-C}) holds.
\end{proof}

\end{document}